\documentclass[a4paper,12pt]{article}
\usepackage[utf8]{inputenc}
\usepackage[T1]{fontenc}
\usepackage{lmodern}
\usepackage{url}
\usepackage{geometry}
\usepackage{amsmath,amsthm,amsfonts,amssymb, mathtools}
\usepackage{graphicx}
\usepackage[colorinlistoftodos]{todonotes}
\usepackage{dcolumn}
\usepackage{enumerate, enumitem}
\usepackage{natbib}
\usepackage[english]{babel}
\usepackage[autostyle]{csquotes}

\usepackage[normalem]{ulem}
\usepackage{xspace}
\usepackage{bbm,comment}
\usepackage[flushleft]{threeparttable}

\newcommand{\E}{\mathbb{E}}          % expectation
\newcommand{\R}{\mathbb{R}}          % real numbers
\newcommand{\p}{\mathbb{P}}          % probability

\newcommand{\indep}{\perp \!\!\! \perp}

\DeclareMathOperator{\Var}{Var}

\DeclareMathOperator{\cov}{cov}

\DeclareMathOperator{\supp}{supp}

\theoremstyle{definition}
\newtheorem{definition}{Definition}
\theoremstyle{plain}
\newtheorem{lemma}{Lemma}

\newtheorem{assumption}{Assumption}
\theoremstyle{remark}
\newtheorem{remark}{Remark}

\newtheorem{prop}{Proposition}
\newtheorem{theorem}{Theorem}
\DeclareMathAlphabet{\mathpzc}{OT1}{pzc}{m}{it}
\title{Distributional Instruments: Identification and Estimation with Quantile Least Squares}
\author{
	Rowan Cherodian\\
	\textit{University of Sheffield}
	\and
	Guy Tchuente\\
	\textit{Purdue University and GLO}
}

\date{\today}

\begin{document}
\maketitle
\begin{abstract}
We study instrumental-variable designs where policy reforms strongly shift the \textit{distribution} of an endogenous variable but only weakly move its \textit{mean}. We formalize this by introducing \textit{distributional relevance}: instruments may be ``purely distributional" with $\Var(\E[X\mid Z])=0$ while $F_{X\mid Z}(\cdot\mid Z)$ varies nontrivially. Within a triangular model, distributional relevance suffices for nonparametric identification of average structural effects via a control function constructed from $F_{X\mid Z}$. We then propose Quantile Least Squares (Q–LS), which aggregates conditional quantiles of $X$ given $Z$ into an optimal mean-square predictor and uses this projection as an instrument in a linear IV estimator. We establish consistency, asymptotic normality, and the validity of standard 2SLS variance formulas, and we discuss regularization across quantiles. Monte Carlo designs show that Q–LS delivers well-centered estimates and near-correct size when mean-based 2SLS suffers from weak instruments. In Health and Retirement Study data, Q–LS exploits Medicare Part D–induced distributional shifts in out-of-pocket risk to sharpen estimates of its effects on depression.\\


\medskip


\end{abstract}

\newpage

\section{Introduction}

A central promise of public health insurance is \emph{insurance}, not just subsidies: it should protect households against the risk of large, unpredictable medical expenses. In the United States this promise is especially salient for older adults. Nearly all Americans over age 65 are enrolled in Medicare, yet many still face substantial exposure to out-of-pocket (OOP) medical spending, particularly for prescription drugs.\footnote{See, for example, \citet{finkelstein2008did} on the initial impact of Medicare on OOP spending and \citet{jones2018lifetime} on lifetime medical spending of retirees.} The introduction of Medicare Part D in 2006 was explicitly designed to reduce this financial risk by expanding subsidized prescription drug coverage and introducing catastrophic protection. A large empirical literature has since documented how Part D and related coverage changes affect drug utilization, adherence, and OOP spending \citep[e.g.][]{ketcham2008medicare,engelhardt2011medicare,park2017medicare,carvalho2019impact}, and a parallel literature shows that Medicare eligibility and supplemental coverage sharply reduce the probability of catastrophic medical expenditures and financial strain among the elderly \citep[e.g.][]{barcellos2015effects,scott2021assessing,jones2019predicting}. 

Much of this work treats financial exposure itself---for example, the level, variance, or tail probability of OOP spending---as an \emph{outcome} of policy. From a policy design perspective, however, we often want the reverse object: the \emph{causal effect} of financial risk exposure on behavior and welfare. Does facing a higher risk of catastrophic prescription drug expenses change adherence, portfolio choice, or retirement decisions? Does improved financial risk protection from Part D or supplemental policies translate into better mental health or reduced financial distress?\footnote{See \citet{ayyagari2015does} and \citet{ayyagari2016medicare} for related evidence on mental health and portfolio choice.} Answering such questions requires treating a measure of risk exposure---for example, a catastrophic-spending indicator or a predicted OOP risk index---as an endogenous regressor and using policy variation as an instrument.

In practice, this approach runs into a familiar econometric obstacle. The most natural instruments for a risk index are policy indicators (e.g.\ Part D eligibility or plan generosity) and features of the plan menu. These instruments primarily shift the \emph{distribution} of OOP spending---compressing right tails, changing dispersion, and altering the frequency of catastrophic events---while often having relatively modest effects on the \emph{mean} of the scalar risk measure itself. Empirically, first-stage regressions of a single risk index on these policy variables can have very low $F$-statistics, even when the underlying policy clearly reshapes the full distribution of spending \citep[see, for example,][]{engelhardt2011medicare,barcellos2015effects,scott2021assessing}. \footnote{ Table \ref{tab:mean_vs_risk_lit} summarizes some representative studies showing mean effect of insurance on mean OOP spending. }Standard IV diagnostics based on mean shifts therefore label the design as ``weak'' and push researchers toward reduced-form analyses, regression discontinuity designs at age~65, or structural models of plan choice \citep[e.g.][]{abaluck2011choice}, rather than causal IV estimates for the effect of financial risk on outcomes.

This paper starts from the observation that, in such settings, the instruments are not weak in any substantive sense: they are \emph{strongly} relevant for the \emph{distribution} of the endogenous variable, even if they appear weak for its mean. We formalize this idea by replacing the classical notion of \emph{mean relevance}---variation in $\E[X\mid Z]$---with a weaker and more general notion of \emph{distributional relevance}: the conditional distribution $F_{X\mid Z}(\cdot\mid Z)$ of the endogenous variable $X$ shifts with the instrument $Z$ in an $L^2$ sense, even when $\E[X\mid Z]$ is constant. In the Medicare example, $X$ may be a scalar index of financial exposure constructed from the OOP distribution, while $Z$ captures policy-induced variation in coverage generosity. Part D and related reforms satisfy distributional relevance whenever they alter the shape (tails, variance, skewness) of the spending distribution, regardless of whether they strongly move its mean.

Our first contribution is to show that distributional relevance is a meaningful notion of instrument strength, not just a descriptive property. Within a non-separable triangular model, we characterize a class of \emph{purely distributional instruments}: variables $Z$ for which $\Var(\E[X\mid Z])=0$ but $F_{X\mid Z}(\cdot\mid Z)$ is non-degenerate in $L^2$. Building on the control-function framework of \citet{cfnsv20_cfa}, we prove that such instruments still identify average structural effects via the control variable $V=F_X(X\mid Z)$, and we relate this notion to existing concepts of instrument strength based on heteroskedasticity \citep{lewbel1997constructing}.
 Thus, even when policy variables are ``purely distributional'' for a risk index, they can support point identification of its causal effect on outcomes.

Our second, and main, contribution is to develop a practical IV estimator tailored to these designs, which we call \emph{Quantile Least Squares} (Q--LS).\footnote{Our notion of ``distributional instruments" is orthogonal to the shift-share (Bartik) literature (see \citep{adao2019shift, borusyak2025practical, goldsmith2020bartik} for a detailed review). There, the focus is on constructing a composite instrument from many shocks and exposure shares. Here, we take the instrument $Z$ as given and show how to optimally exploit its distributional impact on $X$ via conditional quantiles. In particular, Q–LS does not impose a shift–share structure on $Z$; it can be applied whether $Z$ is a simple policy indicator or a more complex shock.} Instead of summarizing the first stage with a conditional mean, Q--LS builds instruments from the \emph{entire} conditional quantile process of $X$ given $Z$. Concretely, we: (i) estimate a series of first-stage quantile regressions of $X$ on $Z$; (ii) form a data-driven linear combination of these conditional quantiles that best predicts $X$ in mean square error; and (iii) use this optimal quantile-aggregated prediction as the instrument in an otherwise standard linear IV estimator. Under distributional relevance and a mild non-degeneracy condition, the resulting optimal Q--LS instrument is always relevant in the classical sense, even when the conditional mean of $X$ is flat in $Z$. In a ``best-case'' world where $X\mid Z$ is a pure location shift, we show that Q--LS collapses to optimal 2SLS and attains the same efficiency bound. In more general cases, Q--LS remains reliable precisely when mean-based instruments become weak.

Methodologically, we characterize the asymptotic properties of Q--LS in a linear triangular model with a single endogenous regressor. We show that (i) under distributional relevance and mild regularity conditions, the Q--LS estimator is consistent and asymptotically normal; (ii) standard heteroskedasticity-robust 2SLS standard errors computed with the generated Q--LS instrument are valid for inference; and (iii) ill-posedness in the optimal-weight problem can be handled with ridge or LASSO regularization across quantiles, yielding stable finite-sample performance. We propose simple diagnostics and a decision chart that place Q--LS alongside conventional weak-IV tools: researchers can test for mean relevance, test for distributional relevance, and then choose between classical 2SLS and Q--LS (and between its regularized variants) based on these diagnostics. In this sense, Q--LS is designed to \emph{complement}, not replace, the existing 2SLS toolkit.

Our Monte Carlo experiments are calibrated to the Medicare context and compare Q--LS to conventional 2SLS and a control-function estimator in designs where instruments primarily shift the distribution of risk rather than its mean. In designs with strong mean relevance, Q--LS and 2SLS deliver nearly identical point estimates and standard errors, confirming that Q--LS does not sacrifice efficiency when classical IV works well. In designs where traditional first-stage $F$-statistics indicate severe weak-instrument problems, Q--LS produces well-centered estimates and near-correct coverage, while 2SLS exhibits substantial bias and undercoverage. Regularized Q--LS variants further stabilise performance when the quantile dictionary is rich relative to sample size.

Finally, we illustrate the usefulness of distributional instruments in the Medicare Part D setting. Using the Health and Retirement Study (HRS), we construct individual-level measures of OOP medical spending risk in real (2015) dollars and exploit Part D as a source of exogenous variation in risk exposure. We show, first, that standard scalar risk measures, such as mean OOP spending or catastrophic-spending indicators, can render an otherwise informative policy design ``weak'' by classical first-stage diagnostics, even when the upper tail of the OOP distribution clearly shifts. Second, we demonstrate how Q--LS, built from a finite dictionary of conditional quantiles of OOP spending, recovers the same causal effects as linear 2SLS in specifications where the real OOP first stage is strong, while delivering much tighter confidence intervals and more stable estimates in specifications where mean relevance is weak. Substantively, our estimates indicate that higher OOP risk is associated with worse mental health among older Americans. To our knowledge, this is the first empirical illustration of how distributional instruments can materially improve the precision and credibility of IV estimates in applied work on medical spending risk.

\paragraph{Contributions.}
Our findings contribute to two main parts of the instrumental‐variables literature. 
First, on the identification side, building on the control–function results of \citet{cfnsv20_cfa}, 
we highlight that their argument continues to apply even when instruments are \emph{purely distributional}: 
even when $\Var(\E[X\mid Z]) = 0$ and the conditional mean of the endogenous variable is completely flat in the 
instrument, nontrivial variation in the conditional distribution $F_{X\mid Z}(\cdot\mid Z)$ is sufficient to identify 
average structural effects in a nonseparable triangular model via a control function based on 
$F_{X\mid Z}(X\mid Z)$. Our contribution is to make this “purely distributional’’ case explicit and to connect it 
to weak–IV diagnostics in linear models, where instruments may look weak in the mean yet be strong in the 
distribution. This formalizes a notion of \emph{distributional relevance} that extends classical mean relevance in a 
direction that is natural for risk and tail–outcome applications.

Second, on the linear IV side, we develop a simple, implementable estimator, Quantile Least Squares (Q--LS), which constructs an optimal distributional instrument from conditional quantiles and then uses it in a conventional linear IV step. Under a strong-distributional-relevance condition, Q--LS shares the same large-sample behavior as 2SLS when instruments are strong in the mean, but in designs where policy reforms mainly reshape the distribution of risk it can extract a much stronger effective first stage from the same variation, complementing---rather than replacing---existing weak-IV diagnostics and robust inference tools. We characterize its finite–sample behavior through Monte Carlo designs.

Our empirical application illustrates how Q--LS can extract a strong
first stage from policy variation that looks weak in classical
mean-based diagnostics, while collapsing back to conventional IV when the instrument is strong in the mean.

The remainder of the paper is organized as follows. Section \ref{sec:lit_rev} discusses weak identification, quantile regression, and some existing ways to circumvent weak identification. Section \ref{sec:setup-DR} introduces the linear IV setup and formalizes distributional relevance. Section~\ref{sec:QLS} defines the Q--LS estimator, presents its asymptotic properties, and compares Q--LS to classical 2SLS and shows that they coincide in a Gaussian/location setting. Section~\ref{sec:implementation} discusses regularization, diagnostics, and practical implementation. Section~\ref{sec:sim} outlines the simulation designs and results. Section~\ref{sec:empiricalapp} presents the Medicare Part D application. Section~\ref{sec:conc} concludes.

\section{Literature Review}\label{sec:lit_rev}

The aim of this research is to develop a general framework that circumvents the classical notion of weak identification of average structural functions in an instrumental variables design. In this section, we first describe and characterize weak identification in triangular models. Second, we summarize quantile regression and discuss how it has been used to address endogeneity in structural functions. Finally, we discuss how our proposed method relates to the existing literature that imposes strong functional form restrictions on the underlying model to circumvent the problem of weak identification.

\subsection{Weak identification in triangular and linear IV models}

Let $(Y,X,Z)$ be random variables that follow a triangular system
\[
    X = g(Z,\eta), \qquad
    Y = f(X,\varepsilon), \qquad
    (\eta,\varepsilon)\perp Z,
\]
where $f$ and $g$ are unknown measurable structural functions, and
$\varepsilon,\eta$ are unobserved disturbances. The vector $Z$ collects the
instruments. The average structural function (ASF) is
\[
  \mu(x) := \E\big[f(x,\varepsilon)\big],
\]
so identification of $\mu(\cdot)$ hinges on how variation in $Z$ propagates to
variation in $X$ through $g(\cdot)$ and, ultimately, to variation in $Y$.

In nonseparable triangular models, identification of $\mu(\cdot)$ often proceeds
via control-function or distribution-regression arguments that construct a
control variable $V=V(X,Z)$ (for example $V=F_{X\mid Z}(X\mid Z)$) such that
$Y \perp Z \mid (X,V)$ and suitable completeness conditions hold; see, among
others, \citet{newey2003instrumental,newey2021control}, \citet{florens2008identification},
\citet{imbens2009identification}, \citet{cfnsv20_cfa}, and \citet{dunker2023nonparametric}.
In this framework, \emph{weak identification} arises when the mapping from $Z$
to the distribution of $X$ is too ``flat’’ to generate informative variation in
the control function. Intuitively, if changes in $Z$ induce only very small or
nearly collinear shifts in $F_{X\mid Z}(\cdot\mid z)$, then the associated
conditional moment conditions are nearly redundant and the inverse problem that
recovers $\mu(\cdot)$ becomes ill–posed or nearly unidentified. This is the
nonlinear analogue of weak instruments: the conditional distribution of $X$
given $Z$ is only weakly sensitive to $Z$, so the control function has little
effective variation.

Closest in spirit to our work is the recent ``Distributional Instrumental Variable (DIV)” method of \citet{holovchak2025div}. They also exploit distributional shifts in $X\mid Z$ and show how to identify and estimate the \emph{entire} interventional distribution of the outcome using flexible generative models in a nonlinear IV setting. DIV provides conditions under which the interventional distribution is identified, and illustrates designs where 2SLS fails but distributional information suffices for identification and accurate estimation of mean and quantile treatment effects. Relative to this approach, we focus on average structural effects in a triangular model and on the construction of an optimal \emph{scalar} instrument for linear IV estimands. Our Q--LS procedure can be viewed as a low-dimensional, analytically tractable counterpart to DIV: it concentrates the distributional information in a finite-dimensional quantile dictionary and delivers an optimal $L^2$-projection that can be used within the classical IV toolkit. Moreover, when the endogenous regressor is (or can be reduced to) a binary treatment, the Q--LS index provides a scalar monotone instrument, so that the resulting 2SLS coefficient admits a standard LATE interpretation as a \emph{positive-weighted} average of complier effects across index increments.

To connect with the more familiar linear IV setting, consider the linear
special case
\[
  Y = X\beta + W'\gamma + u, \qquad
  X = Z'\pi + W'\delta + v,
\]
where $Z$ are excluded instruments, $W$ are included exogenous regressors, and
$\E[u\mid Z,W]\neq 0$ while $\E[v\mid Z,W]=0$. Classical weak-instrument
problems arise when the first-stage coefficients $\pi$ are small, so that the
concentration parameter is close to zero and the instruments explain only a
tiny fraction of the variation in $X$; see \citet{staiger1997instrumental},
\citet{stock2005asymptotic}, and \citet{dufour1997some} for canonical analyses.
In this case, 2SLS and related estimators exhibit large finite-sample bias
toward the OLS estimand, heavy-tailed or multimodal sampling distributions, and
nonstandard asymptotic behavior \citep[see][for a general overview]{andrews2014weak,andrews2019weak}. These issues are
amplified in many-instrument settings \citep[e.g.][]{bekker1994alternative,hansen2008estimation}, where the
effective first-stage signal per instrument becomes small even if the joint $F$-statistic is sizable.

The weak-identification phenomena in triangular models are conceptually similar
but operate at the level of \emph{distributions} rather than means. Instead of
small first-stage coefficients on $Z$ in a linear projection of $X$ on $Z$, the
problem is that the instrument-induced shifts in $F_{X\mid Z}(\cdot\mid z)$ are
too limited or nearly fail completeness, so the mapping from structural objects
to observables is nearly singular. As emphasized by \citet{l19_id_zoo}, both
linear and nonlinear models admit a continuum of identification strengths,
ranging from strong point identification with regular asymptotics, through
weak or set identification with nonstandard limits, to complete lack of
identification.

Our Q--LS estimator can be viewed as a two-step device that first maps the
original instrument vector $Z$ into a scalar \emph{distributional instrument}
\[
  h(Z) \in \mathcal{H} = \Big\{ \int_0^1 \omega(\tau) Q_{X\mid Z}(\tau\mid Z)\,d\tau
  : \omega \in L^2(0,1)\Big\},
\]
and then applies conventional linear IV using $h(Z)$ in place of $Z$. From this
perspective, Assumption~\ref{ass:QLS-relevance} is the analogue of a strong-IV
condition on the \emph{generated} instrument $h(Z)$: it rules out sequences of
designs in which the covariance between $X$ and $h(Z)$ drifts to zero at a
$\sqrt{n}$–local rate.

We therefore do \emph{not} claim that Q--LS is generally robust in the weak-IV
sense of \citet{staiger1997instrumental} or that it uniformly dominates
Anderson--Rubin or conditional tests \citep[e.g.][]{andrews2019weak,olea2013robust} when instruments are arbitrarily weak. Rather, the contribution of our \emph{distributional relevance} framework is to highlight a different
margin of strength: designs in which $\Var(\E[X\mid Z])$ is small or zero, so
that instruments appear weak for the \emph{mean}, may still be strong for the
\emph{distribution}, with $\E[h_{\mathrm{opt}}(Z)^2]$ bounded away from zero.
In such cases, Q--LS leverages distributional variation in $X\mid Z$ to
construct an effective instrument $h(Z)$ that satisfies a strong-IV condition
even when the original mean-based instrument fails conventional first-stage
diagnostics. Our notion of distributional relevance thus specifies a
``strong-instrument’’ regime at the distributional level and complements both
the generative DIV approach and the classical weak-IV toolkit. This allows
us to focus on designs where instruments may be weak for linear first stages
but remain informative about the \emph{distribution} of $X$, and hence can
still support identification and efficient estimation once the first stage is
recast in distributional terms.

\subsection{Quantile regression and endogeneity}

Quantile regression was first formalised as an optimization problem minimizing asymmetrically weighted absolute residuals by \cite{kb78_qr}, to estimate any chosen quantile $\tau\in(0,1)$ of the distribution of an outcome variable $X$  conditional on a set of exogenous regressors $Z$,

\begin{equation*}
Q_{X|Z}\big(\tau\big|Z=z\big)=\inf \{X:F_{X|Z}(x|z)\geq \tau\},
\end{equation*}
where $F_{X|Z}(x|z)$ is the conditional cumulative distribution function (CDF) of $X$ given $Z=z$. 

The linear quantile regression model specifies

\begin{equation*}
Q_X\big(\tau\big|Z=z\big)=Z'\pi(\tau),
\end{equation*}

where $\pi(\tau)$ is estimated by solving the expected check-loss function

\begin{equation*}
\min_\pi \E\big[\rho_\tau(X-Z'\pi)\big], \;\;\; \rho_\tau(u)=u\big(\tau-\mathbbm{1}\{u<0\}\big)
\end{equation*}

This objective function is convex but non-differentiable, yielding a linear programming problem that can be solved efficiently in moderately large samples.

Identification of quantile regression relies on mild regularity conditions, including correct specification of the conditional quantile function, sufficient variation in the regressors, and a rank condition ensuring uniqueness of the conditional quantile. These conditions do not require parametric assumptions on the error distribution

\cite{aai02_qiv} were the first to combine quantile methods with instrumental variables to estimate distributional treatment effects—that is, how a binary treatment affects quantiles of the outcome distribution—under endogeneity. \cite{ch05_qiv} generalized this approach to estimate treatment effects at different points of the outcome distribution (quantile treatment effects) under endogeneity and introduced the notion of the structural quantile function. See \cite{w20} for a detailed comparison of \cite{ch05_qiv} and \cite{aai02_qiv}. Since these seminal contributions, many variants of quantile instrumental variable methods have been proposed. For example, \cite{p20_uncon_qiv} develop an instrumental-variable approach to estimate unconditional quantile treatment effects, while \cite{clp16_gqiv} propose methods for group-level quantile treatment effects.

There is also a growing literature that integrates over a grid of quantiles to estimate endogenous structural functions. \cite{cfk15_qiv} propose a Censored Quantile Instrumental Variables (CQIV) estimator that uses a control-function approach to estimate structural quantiles of a censored outcome with endogenous regressors. In their framework, the control function is estimated by integrating over a grid of quantiles. \cite{cfnsv20_cfa} extend the control-function framework of \cite{cfk15_qiv} and develop a unified approach for estimating average, quantile, and distributional structural functions. The control-function estimator we propose for nonseparable triangular models estimates the control function using the same strategy as \cite{cfnsv20_cfa}. See Appendix \ref{sec:app_nstm} for further details.

\subsection{Existing ways to circumvent weak identification}

There is a substantial literature that seeks to circumvent the problem of weak identification by imposing specific functional form restrictions. In this section, we discuss several key strands of this literature and how it relates to our proposed notion of distributional relevance.

There is a growing literature that imposes functional form restrictions on higher moments. \citet{lewbel1997constructing} show that, in linear models with measurement error and endogeneity, instruments can be constructed from second-moment variation:
\[
    Z_i^{\text{Lev}} = (Z_i - \bar Z)(X_i - \bar X),
\]
provided that the variation of $X$ conditional on $Z$ changes with $Z$ even when
$\E[X\mid Z]$ is constant. This identification strategy exploits \emph{variance shifts}. Several studies have extended \citet{lewbel1997constructing} and used second-moment restrictions to achieve identification under weaker functional form restrictions; see, for example, \cite{kv10_iv_cfa_hetro,l12_hetro}. These restrictions can be viewed as a specific form of distributional relevance (variance shifts) and are therefore encompassed by our framework.

\citet{t23_iv_icm} also impose a functional form restriction that ensures identification in models of the form $\E[Y - X\beta \mid Z]=0$ without excluded instruments. Specifically, \citet{t23_iv_icm} assume a nonlinear completeness condition requiring injectivity of the conditional mean operator:

\[
    \E[X\mid Z] \tau = 0
    \quad\Rightarrow\quad \tau = 0.
\]

Completeness relies on sufficiently rich nonlinear mean dependence of $X$ conditional on $Z$, that is, on the mapping $X \mapsto \E[X\mid Z]$.

Our approach is fundamentally different from \citet{t23_iv_icm}. Identification does not rely on the conditional \emph{mean} operator or its injectivity. Instead, distributional variation in $X$ suffices to identify the average structural function $\mu(x)$, even when $\E[X\mid Z]$ is constant and completeness fails. Consequently, our results do not require completeness, mean dependence, or any functional-form restrictions of the type imposed by \citet{t23_iv_icm}.

\citet{babii2020completeness} study estimation in nonidentified linear inverse
problems of the from $K\varphi = r$, where $K$ is a compact operator between Hilbert spaces. In nonparametric IV models, $K$ corresponds to the conditional expectation operator $(K\varphi)(z)=\E[\varphi(X)\mid Z=z]$. When completeness (injectivity of $K$) fails, $\varphi$ is not point-identified, but spectral regularization methods converge to the ``best approximation'' $\varphi_1$ in the orthogonal complement of the null space of $K$. Their work characterizes risk bounds and asymptotic distributions under varying degrees of identification.

Our framework differs in two central respects. First, identification is not
formulated as an inverse problem with operator $K\varphi=r$. Instead, we
use a structural triangular model with distributional relevance, which yields full nonparametric identification of the average structural function $\mu(x)$ when the usual  nonparametric IV completeness condition fails. Second, whereas
\citet{babii2020completeness} obtain convergence to a best approximation under
nonidentification, our results deliver point identification of $\mu(x)$ by
exploiting the rank structure of the first stage and the distributional
variation of $X$. Consequently, our contribution is complementary: we
provide an identification and estimation framework that does not rely on
operator injectivity or completeness, and remains applicable with
instruments that are mean-irrelevant but distributionally strong.

\section{Linear IV setup and distributional relevance}

\label{sec:setup-DR}

\subsection{Linear IV model}

We consider the linear triangular model
\begin{equation}
  Y = \alpha + \beta  X + Z_1'\gamma + \varepsilon,
  \qquad \E[\varepsilon \mid Z] = 0,
  \label{eq:structural}
\end{equation}
where $Y$ is the outcome, $X$ is a scalar endogenous regressor, $Z_1$ is a
vector of exogenous covariates included in the outcome equation, and
$Z = (Z_1, Z_2)$ collects $Z_1$ and a vector of excluded instruments
$Z_2$.

In the classical IV framework, instrument strength is measured by variation
in the conditional mean of $X$ given $Z$:
\[
  m(Z) := \E[X\mid Z], \qquad \Var(m(Z)) > 0.
\]
When $\Var(m(Z))$ is close to zero, instruments are labelled ``weak'' and
standard 2SLS inference is unreliable.

Our aim is to replace this mean-based notion of strength with a weaker and
more general concept based on the \emph{full conditional distribution} of
$X$.

\subsection{Distributional relevance}

Let
\[
  F_{X\mid Z}(x\mid Z)
  :=
  \p(X \le x \mid Z),
  \qquad
  \bar F_X(x) := \E\big[ F_{X\mid Z}(x\mid Z)\big]
\]
denote the conditional and unconditional distribution functions of $X$. We
measure variation in $F_{X\mid Z}(\cdot\mid Z)$ using an $L^2$ norm over
$x$. For concreteness, let $\lambda$ be Lebesgue measure on a compact
interval containing the support of $X$, and define
\[
  \| F_{X\mid Z}(\cdot\mid Z) \|_{L^2(\lambda)}^2
  :=
  \int \big( F_{X\mid Z}(x\mid Z) \big)^2\,d\lambda(x).
\]

\begin{definition}[Mean relevance]
\label{def:mean-relevance}
We say that the instrument $Z$ is \emph{mean-relevant} for $X$ if
\[
  \Var\big( \E[X \mid Z] \big)
  = \E\Big[\big(\E[X\mid Z] - \E[X]\big)^2\Big] > 0.
\]
\end{definition}

\begin{definition}[Distributional relevance]
\label{def:dist-relevance}
We say that the instrument $Z$ is \emph{distributionally relevant} for $X$
if
\begin{equation}
  \E\Big[
    \big\| F_{X\mid Z}(\cdot \mid Z) - \bar F_X(\cdot) \big\|_{L^2(\lambda)}^2
  \Big]
  > 0.
  \label{eq:dist-relevance-L2}
\end{equation}
Equivalently, there exist $z_p,z_q$ in the support of $Z$ such that
\[
  \big\| F_{X\mid Z}(\cdot \mid z_p) - F_{X\mid Z}(\cdot \mid z_q)
  \big\|_{L^2(\lambda)} > 0.
\]
\end{definition}

Any change in $\E[X\mid Z]$ across $Z$ induces a difference in
$F_{X\mid Z}(\cdot\mid Z)$, so mean relevance implies distributional
relevance. The converse need not hold: $Z$ can be distributionally relevant
even when $\E[X\mid Z]$ is constant in $Z$. This motivates the following
concept.

\begin{definition}[Purely distributional instruments]
\label{def:purely-dist}
We say that $Z$ is a \emph{purely distributional instrument} for $X$ if
\begin{enumerate}[label=(\roman*)]
  \item $Z$ is distributionally relevant for $X$ in the sense of
  Definition~\ref{def:dist-relevance}, and
  \item $Z$ is mean-irrelevant for $X$, i.e.
  \[
    \Var(\E[X\mid Z]) = 0
    \quad\Longleftrightarrow\quad
    \E[X\mid Z=z] = \E[X]
    \quad \text{for all } z.
  \]
\end{enumerate}
\end{definition}

Classical IV analysis treats \emph{mean relevance}---variation in
$\E[X\mid Z]$---as the benchmark notion of instrument strength.
Distributional relevance weakens this requirement: it allows
$\E[X\mid Z]$ to be completely flat in $Z$, as long as the conditional
distribution $F_{X\mid Z}(\cdot\mid Z)$ varies nontrivially in $L^2$.
In this sense, an instrument can be ``strong'' for the \emph{distribution}
of $X$ even when it is ``weak'' for its mean.
We refer to instruments that satisfy distributional relevance but have
$\Var(\E[X\mid Z])=0$ as \emph{purely distributional}: they convey no
information about $\E[X\mid Z]$ but still transmit rich information
about the risk environment through higher moments and tail behavior.

In the linear IV analysis that follows, distributional relevance and the
linear quantile representation for $X\mid Z$ guarantee the existence of
an optimal Q--LS instrument that is relevant in the classical sense
unless the projection degenerates.
In the nonseparable triangular model in Appendix~\ref{sec:app_nstm},
the condition that $Z$ is purely distributional is enough for the
identification of the average structural function via a control
function based on $F_{X\mid Z}(X\mid Z)$.

\begin{remark}[Illustrative example]
\label{rem:purely-dist-example}
Let $Z_2\in\{0,1\}$ be a binary instrument with $\p(Z_2=0)=\p(Z_2=1)=1/2$,
and let $U\sim\mathcal{N}(0,1)$ be independent of $Z_2$. Define
\[
  X = \sigma(Z_2)\,U,
  \qquad
  \sigma(0)=1,\ \sigma(1)=2.
\]
Then $\E[X\mid Z_2=z_2]=0$ for all $z_2$, so $\Var(\E[X\mid Z_2])=0$ and
$Z_2$ is mean-irrelevant. However, the conditional distributions differ:
$X\mid Z_2=0\sim\mathcal{N}(0,1)$ and $X\mid Z_2=1\sim\mathcal{N}(0,4)$, so
$F_{X\mid Z_2}(\cdot\mid 0)\neq F_{X\mid Z_2}(\cdot\mid 1)$ in $L^2$. Thus
$Z_2$ is a purely distributional instrument. In a classical first-stage
regression of $X$ on $Z_2$, the coefficient on $Z_2$ would be zero, but the
instrument carries substantial information about the dispersion of $X$.
\end{remark}

In what follows, distributional relevance will be our primitive assumption
on the instrument. We show that, under a linear quantile representation for
$X\mid Z$, it is enough to construct a strong instrument as a linear
functional of the conditional quantile process, and to obtain consistent and
asymptotically normal estimates of $(\alpha,\beta,\gamma)$ in the
linear model \eqref{eq:structural}.

\section{The Q--LS estimator in the linear model}
\label{sec:QLS}

We now define the Quantile Least Squares IV (Q--LS) estimator. The key
object is an \emph{optimal} instrument constructed as a linear functional of
the conditional quantile process of $X\mid Z$. 

\subsection{Conditional quantiles and quantile-aggregated instruments}

We work with the conditional quantile function of $X$ given $Z$. For each
$\tau\in(0,1)$, assume
\begin{equation}
  Q_{X\mid Z}(\tau\mid Z)
  = Z'\pi_0(\tau),
  \qquad \tau\in(0,1),
  \label{eq:first-stage-quantile}
\end{equation}
where $\pi_0(\tau)\in\R^{\dim(Z)}$ is a measurable function of $\tau$. We
write
\[
  g(\tau,Z) := Q_{X\mid Z}(\tau\mid Z) = Z'\pi_0(\tau).
\]
Under mild regularity, distributional relevance is equivalent to the
statement that the conditional quantile process
$\{g(\tau,Z):\tau\in(0,1)\}$ is non-constant in $Z$ in $L^2$.\footnote{For concreteness, we implement Q--LS using linear quantile
regressions of $X$ on $Z$ at a finite grid of quantile indices.
This choice is purely for illustration: any alternative specification
for the conditional quantiles (e.g., nonlinear, spline or series-based,
or machine–learning first stages) could be used to construct the
dictionary $G(Z)$ and hence the Q--LS instrument.}

We restrict attention to instruments that are linear functionals of this
quantile process. Let $\omega(\cdot)$ be a square-integrable weight function
on $(0,1)$. We define the quantile-aggregated instrument
\begin{equation}
  h_\omega(Z)
  :=
  \int_0^1 \omega(\tau)\, g(\tau,Z)\,d\tau
  =
  \int_0^1 \omega(\tau)\, Q_{X\mid Z}(\tau\mid Z)\,d\tau.
  \label{eq:QLS-instrument}
\end{equation}
The corresponding moment conditions for the structural parameter
$\theta := (\alpha,\beta,\gamma')'$ are
\begin{equation}
  \E\!\left[
    \begin{pmatrix}
      h_\omega(Z)\\
      Z_1
    \end{pmatrix}
    \big(
      Y - \alpha - \beta X - Z_1'\gamma
    \big)
  \right] = 0.
  \label{eq:QLS-moment}
\end{equation}

Let
\[
  \mathcal{H}
  :=
  \Big\{
    h_\omega(\cdot): h_\omega(Z)
    = \int_0^1 \omega(\tau)\,g(\tau,Z)\,d\tau,
    \ \omega\in L^2(0,1)
  \Big\}
\]
be the linear span (in $L^2$) of the conditional quantile process.

\subsection{Optimal Q--LS and relevance under distributional relevance}
\label{subsec:opt-qls}

Following the classical optimal IV logic for a single endogenous regressor,
we consider the element of $\mathcal{H}$ that best predicts $X$ in mean square
error.

\begin{definition}[Optimal Q--LS]
\label{def:opt-QLS}
The optimal Q--LS weight function $\omega_{\mathrm{opt}}(\cdot)$ is any solution to
\begin{equation}
  \omega_{\mathrm{opt}}
  \in \arg\min_{\omega\in L^2(0,1)}
  \E\!\left[
    \big(
      X - h_\omega(Z)
    \big)^2
  \right],
  \label{eq:omega-opt}
\end{equation}
and the associated optimal Q--LS instrument is
\begin{equation}
  h_{\mathrm{opt}}(Z)
  :=
  h_{\omega_{\mathrm{opt}}}(Z)
  =
  \int_0^1 \omega_{\mathrm{opt}}(\tau)\, Q_{X\mid Z}(\tau\mid Z)\,d\tau.
  \label{eq:h-opt}
\end{equation}
\end{definition}

By construction, $h_{\mathrm{opt}}(Z)$ is the $L^2$-projection of $X$ onto
the closed linear span $\mathcal{H}$ of quantile-generated instruments.
The next lemma shows that, whenever this projection is nonzero, the resulting
instrument is automatically relevant for $X$ in the classical sense.

\begin{lemma}[Distributional relevance and optimal Q--LS]
\label{lem:DR-relevance}
Suppose Assumption~\ref{def:dist-relevance} and
\eqref{eq:first-stage-quantile} hold, and let $h_{\mathrm{opt}}$ be defined
as in Definition~\ref{def:opt-QLS}. Then
\[
  X = h_{\mathrm{opt}}(Z) + r(Z),
  \qquad r \perp \mathcal{H},
\]
and either
\[
  h_{\mathrm{opt}}(Z) = 0 \quad \text{a.s.},
\]
or
\[
  \cov\big(h_{\mathrm{opt}}(Z), X\big)
  = \E\big[h_{\mathrm{opt}}(Z)^2\big] > 0.
\]
In particular, whenever $h_{\mathrm{opt}}$ is non-degenerate it is a relevant
instrument for $X$.
\end{lemma}

\begin{proof}
By definition of $h_{\mathrm{opt}}$ as the $L^2$-projection of $X$ onto the
closed linear subspace $\mathcal{H}$, we can write
\[
  X = h_{\mathrm{opt}}(Z) + r(Z),
\]
where $r(Z)$ is the projection residual and satisfies
$\E[h(Z)\,r(Z)] = 0$ for all $h\in\mathcal{H}$. In particular,
$\E[h_{\mathrm{opt}}(Z)\,r(Z)] = 0$, so
\[
  \E[X\,h_{\mathrm{opt}}(Z)]
  = \E\big[(h_{\mathrm{opt}}(Z)+r(Z))\,h_{\mathrm{opt}}(Z)\big]
  = \E\big[h_{\mathrm{opt}}(Z)^2\big]
  + \E\big[h_{\mathrm{opt}}(Z)\,r(Z)\big]
  = \E\big[h_{\mathrm{opt}}(Z)^2\big].
\]
Thus either $h_{\mathrm{opt}}(Z)=0$ a.s., in which case both sides vanish, or
$\E[h_{\mathrm{opt}}(Z)^2]>0$ and $h_{\mathrm{opt}}$ is correlated with $X$.
\end{proof}

Lemma~\ref{lem:DR-relevance} delivers a clean dichotomy: the optimal Q--LS
instrument is either identically zero or, if nonzero, it is automatically
relevant and its ``strength'' is measured by $\E[h_{\mathrm{opt}}(Z)^2]$.
Distributional relevance guarantees that $\mathcal{H}$ contains more than
just constants, but it does \emph{not} by itself rule out the degenerate case
$h_{\mathrm{opt}}\equiv 0$. A simple example is a pure variance-shift design,
\[
  X = \sigma(Z)\,U,
  \qquad \E[U]=0,\quad U \indep Z,
\]
with $\sigma(Z)$ non-constant. In this case the conditional distributions
$F_{X\mid Z}(\cdot\mid Z)$ vary with $Z$ (distributional relevance), yet
$\cov(X,h(Z))=0$ for every $h\in\mathcal{H}$ so that $h_{\mathrm{opt}}(Z)\equiv 0$.
Our Design~B in Section~\ref{sec:sim} provides a Monte Carlo illustration of
this phenomenon.

For identification, we therefore add an explicit relevance condition that is
the exact analogue of the usual mean-relevance condition in classical IV.

\begin{assumption}[Strong distributional relevance]\label{ass:QLS-relevance}
Under the linear quantile representation \eqref{eq:first-stage-quantile}, the optimal Q--LS instrument $h_{\mathrm{opt}}(Z)$ defined in \eqref{eq:h-opt} satisfies
\[
  \E\big[h_{\mathrm{opt}}(Z)^2\big] \ge c > 0
\]
for some constant $c$ that does not depend on the sample size $n$.
\end{assumption}

Assumption~\ref{ass:QLS-relevance} is the distributional analogue of the usual ``strong IV'' condition in linear models, where one assumes that the first-stage coefficient on the excluded instruments does not drift to zero as $n$ grows (e.g.\ \citealp{staiger1997instrumental, stock2005asymptotic}. Here we impose a lower bound on the variance of the optimal Q--LS instrument $h_{\mathrm{opt}}(Z)$, ruling out sequences of designs in which instruments become weak either in the mean or in the distribution. Throughout our asymptotic analysis we work under Assumption~\ref{ass:QLS-relevance}, so our results should be interpreted as \emph{strong-distributional-IV} asymptotics.

Under Assumption~\ref{ass:QLS-relevance}, Lemma~\ref{lem:DR-relevance}
implies that $h_{\mathrm{opt}}$ is a non-degenerate, relevant instrument
constructed purely from the conditional distribution of $X\mid Z$, even when
$\E[X\mid Z]$ is constant in $Z$. 

Throughout the linear IV analysis, we assume that distributional
relevance holds and that $X\mid Z$ admits the linear quantile
representation in \eqref{eq:first-stage-quantile}.
Under these conditions, Lemma~\ref{lem:DR-relevance} shows that the
optimal Q--LS instrument $h_{\mathrm{opt}}(Z)$ is either identically
zero or a classically relevant instrument for $X$, and Assumption~\ref{ass:QLS-relevance}
rules out the degenerate case.

\subsection{Finite-grid implementation}

In practice, we approximate the integral in \eqref{eq:QLS-instrument} using
a finite grid of quantiles. Let
\[
  0 < \tau_1 < \cdots < \tau_K < 1
\]
be a set of quantile indexes with gap of order $1/K$, and let $\Delta\tau_k$
denote the associated quadrature weights (for simplicity, one can take
$\Delta\tau_k = 1/K$).

For each $\tau_k$, we estimate the first-stage quantile regression
\begin{equation}
  \widehat{\pi}_n(\tau_k)
  \in \arg\min_{\pi}
  \frac{1}{n}\sum_{i=1}^n
  \rho_{\tau_k}\big(X_i - Z_i'\pi\big),
  \qquad k=1,\dots,K,
  \label{eq:sample-qr}
\end{equation}
and define the estimated quantile-based instruments
\[
  \widehat{g}_k(Z_i) := Z_i'\widehat{\pi}_n(\tau_k),
  \qquad i=1,\dots,n,\quad k=1,\dots,K.
\]
Stacking them, let
\[
  \widehat{G}_i
  :=
  \big(
    \widehat{g}_1(Z_i),
    \dots,
    \widehat{g}_K(Z_i)
  \big)',
  \qquad
  \widehat{G}
  :=
  \begin{pmatrix}
    \widehat{G}_1'\\
    \vdots\\
    \widehat{G}_n'
  \end{pmatrix}
  \in \R^{n\times K}.
\]

\paragraph{Discrete Q--LS weights.}
A natural discrete analogue of \eqref{eq:omega-opt} chooses weights
$w\in\R^K$ to best predict $X$ from the $K$ quantile instruments:
\begin{equation}
  \widehat{w}_n
  \in \arg\min_{w\in\R^K}
  \frac{1}{n}\sum_{i=1}^n
  \big(
    X_i - \widehat{G}_i' w
  \big)^2,
  \label{eq:w-hat-opt}
\end{equation}
so that, when $\widehat{G}'\widehat{G}$ is invertible,
\begin{equation}
  \widehat{w}_n
  =
  (\widehat{G}'\widehat{G})^{-1}\widehat{G}'X,
  \qquad
  X = (X_1,\dots,X_n)'.
  \label{eq:w-hat-closed-form}
\end{equation}
The corresponding finite-sample Q--LS first-stage prediction is
\begin{equation}
  \widehat{X}_i^{\mathrm{Q\text{--}LS}}
  :=
  \widehat{G}_i'\widehat{w}_n
  =
  \sum_{k=1}^K \widehat{w}_{n,k}\,\widehat{g}_k(Z_i),
  \qquad i=1,\dots,n.
  \label{eq:Xhat-QLS}
\end{equation}
This $\widehat{X}_i^{\mathrm{Q\text{--}LS}}$ is the empirical best linear
predictor of $X_i$ in the span of the $K$ quantile-based instruments.

\paragraph{Second stage.}
Let $y = (Y_1,\dots,Y_n)'$ be the outcome vector, and define the regressor
and instrument matrices
\[
  S = [\mathbf{1},\, X,\, Z_1],
  \qquad
  M = [\widehat{X}^{\mathrm{Q\text{--}LS}},\, Z_1],
\]
where $\widehat{X}^{\mathrm{Q\text{--}LS}} = (\widehat{X}_1^{\mathrm{Q\text{--}LS}},
\dots,\widehat{X}_n^{\mathrm{Q\text{--}LS}})'$ and $\mathbf{1}$ denotes the
$n$-vector of ones. The Q--LS estimator of
$\theta = (\alpha,\beta,\gamma')'$ is the usual 2SLS estimator
\begin{equation}
  \widehat{\theta}_n^{\mathrm{Q\text{--}LS}}
  :=
  \big(S'P_M S\big)^{-1} S'P_M y,
  \qquad
  P_M = M(M'M)^{-1}M'.
  \label{eq:theta-QLS}
\end{equation}

\subsection{Asymptotic properties}
\label{subsec:QLS-asymp}

We now state conditions under which the Q--LS estimator is consistent and
asymptotically normal. The key identification condition is distributional
relevance, which, via Lemma~\ref{lem:DR-relevance}, implies relevance of the
optimal Q--LS instrument.

\begin{assumption}[Regularity and identification]
\label{ass:consistency}
\leavevmode
\begin{enumerate}[label=(\roman*)]
  \item The sample $\{(Y_i,X_i,Z_i)\}_{i=1}^n$ is i.i.d., and
  $\E\|Z\|^2 < \infty$, $\E|X|^2 < \infty$, $\E|Y|^2 < \infty$.
  \item For each $\tau\in(0,1)$, the conditional quantile $Q_{X\mid Z}(\tau\mid Z)$
  exists and satisfies \eqref{eq:first-stage-quantile}, with $\pi_0(\cdot)$
  continuous on $(0,1)$ and $\int_0^1 \|\pi_0(\tau)\|^2 d\tau < \infty$.
  \item The structural error satisfies $\E[\varepsilon\mid Z]=0$ and
  $\E[\varepsilon^2\mid Z] < \infty$.
  \item Distributional relevance holds in the sense of
  Definition~\ref{def:dist-relevance}, and the span $\mathcal{H}$ generated
  by $\{g(\tau,Z):\tau\in(0,1)\}$ is not reduced to constants. Let
  $h_{\mathrm{opt}}$ be defined as in Definition~\ref{def:opt-QLS}, and
  suppose the population matrix
  \[
    \E\big[ S_0(Z)' S_0(Z) \big],
    \qquad S_0(Z) = (1, X, Z_1')
  \]
  is nonsingular.
  \item The quantile grid $\{\tau_k\}_{k=1}^K$ is dense in $(0,1)$ as
  $K=K_n\to\infty$, and the number of quantiles satisfies
  $K_n^2/n \to 0$ as $n\to\infty$.
\end{enumerate}
\end{assumption}

\begin{prop}[Consistency of the Q--LS estimator]
\label{prop:QLS-consistency}
Suppose Assumptions  \ref{ass:QLS-relevance} and \ref{ass:consistency}  hold, and the weights
$\widehat{w}_n$ are defined by \eqref{eq:w-hat-opt} with $K=K_n\to\infty$
and $K_n^2/n\to 0$. Then
\[
  \widehat{\theta}_n^{\mathrm{Q\text{--}LS}}
  \xrightarrow{p} \theta
  \qquad\text{as } n\to\infty.
\]
\end{prop}

\begin{proof}
By Assumption~\ref{ass:consistency}(ii) and standard quantile regression
theory, we have, for each fixed $\tau$,
$\widehat{\pi}_n(\tau) \xrightarrow{p} \pi_0(\tau)$, and, since the grid
$\{\tau_k\}$ becomes dense and $K_n^2/n\to 0$, the array
$\{\widehat{\pi}_n(\tau_k)\}$ uniformly approximates $\{\pi_0(\tau_k)\}$ in
mean square. It follows that
\[
  \max_{1\le k\le K_n}
  \E\big[
    \big(
      \widehat{g}_k(Z) - g(\tau_k,Z)
    \big)^2
  \big] \to 0,
\]
and hence $\widehat{G}'\widehat{G}/n$ converges in probability to the
corresponding population limit, and similarly for $\widehat{G}'X/n$.
Under $K_n^2/n\to 0$, the estimation error from the first-stage quantile
regressions is $o_p(n^{-1/2})$ in the moment conditions.

Therefore the discrete weights $\widehat{w}_n$ converge in probability to
the population weights $w_0$ that minimize
$\E[(X-G_0(Z)'w)^2]$ in the corresponding finite-dimensional approximation
space, where $G_0(Z)$ stacks the $g(\tau_k,Z)$'s. As $K_n\to\infty$ and the
grid becomes dense, the discrete approximation $G_0(Z)'w_0$ converges in
$L^2$ to the continuum optimal instrument $h_{\mathrm{opt}}(Z)$ defined in
Definition~\ref{def:opt-QLS}. In particular,
\[
  \E\big[
    \big(
      \widehat{X}_i^{\mathrm{Q\text{--}LS}} - h_{\mathrm{opt}}(Z_i)
    \big)^2
  \big] \to 0.
\]
Standard arguments for 2SLS with generated instruments then imply that
$\widehat{\theta}_n^{\mathrm{Q\text{--}LS}}$ converges in probability to the
unique solution of the population IV moment condition with instrument
$h_{\mathrm{opt}}(Z)$, which is $\theta_0$ by Assumption~\ref{ass:consistency}(iv).
\end{proof}

For asymptotic normality, let
\[
  S_i := \big(1, X_i, Z_{1i}'\big)',
  \qquad
  M_i := \big(\widehat{X}_i^{\mathrm{Q\text{--}LS}}, Z_{1i}'\big)'
\]
denote the stacked regressor and instrument vectors used in the second
stage, and let
\[
  \widehat{\varepsilon}_i
  :=
  Y_i - S_i'\widehat{\theta}_n^{\mathrm{Q\text{--}LS}}.
\]
Define the sample matrices
\[
  \widehat{A}_n
  :=
  \frac{1}{n}\sum_{i=1}^n S_i M_i',
  \qquad
  \widehat{Q}_n
  :=
  \frac{1}{n}\sum_{i=1}^n M_i M_i',
\]
and the heteroskedasticity-robust ``meat'' matrix
\[
  \widehat{\Omega}_n
  :=
  \frac{1}{n}\sum_{i=1}^n
    \widehat{\varepsilon}_i^{\,2}\,
    M_i M_i'.
\]

\begin{assumption}[Additional conditions for asymptotic normality]
\label{ass:an}
In addition to Assumption~\ref{ass:consistency}, suppose that:
\begin{enumerate}[label=(\roman*)]
  \item The fourth moments of $(Y,X,Z)$ are finite:
  $\E\|Z\|^4 < \infty$, $\E|X|^4 < \infty$, $\E|Y|^4 < \infty$.
  \item The conditional fourth moment of the structural error is finite:
  $\E[\varepsilon^4\mid Z] < \infty$ a.s.
  \item The eigenvalues of $\E[M_i^0 M_i^{0\prime}]$ are bounded away
  from zero and infinity, where $M_i^0 = (h_{\mathrm{opt}}(Z_i),Z_{1i}')'$ is
  the population instrument vector associated with the optimal Q--LS, and
  the same holds for $\E[S_i S_i']$.
\end{enumerate}
\end{assumption}

Let
\[
  A
  :=
  \E[S_i M_i^{0\prime}],
  \qquad
  Q
  :=
  \E[M_i^0 M_i^{0\prime}],
  \qquad
  \Omega
  :=
  \E\big[\varepsilon_i^{2} M_i^0 M_i^{0\prime}\big].
\]
The population asymptotic covariance matrix of the Q--LS estimator can then
be written in the usual IV ``sandwich'' form
\begin{equation}
  V_{\mathrm{Q\text{--}LS}}
  :=
  \big(A Q^{-1} A'\big)^{-1}
  \big( A Q^{-1} \Omega Q^{-1} A' \big)
  \big(A Q^{-1} A'\big)^{-1}.
  \label{eq:V-QLS}
\end{equation}

\begin{prop}[Asymptotic normality of the Q--LS estimator]
\label{prop:QLS-an}
Suppose Assumptions \ref{ass:QLS-relevance}, \ref{ass:consistency} and \ref{ass:an} hold. Then
\[
  \sqrt{n}\big(
    \widehat{\theta}_n^{\mathrm{Q\text{--}LS}} - \theta
  \big)
  \ \xrightarrow{d}\
  \mathcal{N}\big(0, V_{\mathrm{Q\text{--}LS}}\big),
\]
where $V_{\mathrm{Q\text{--}LS}}$ is given in \eqref{eq:V-QLS}. Moreover, the
heteroskedasticity-robust 2SLS covariance matrix computed with the generated
Q--LS instrument,
\begin{equation}
  \widehat{V}_n^{\mathrm{Q\text{--}LS}}
  :=
  \big(\widehat{A}_n \widehat{Q}_n^{-1} \widehat{A}_n'\big)^{-1}
  \big(\widehat{A}_n \widehat{Q}_n^{-1} \widehat{\Omega}_n
        \widehat{Q}_n^{-1} \widehat{A}_n'\big)
  \big(\widehat{A}_n \widehat{Q}_n^{-1} \widehat{A}_n'\big)^{-1},
  \label{eq:Vhat-QLS}
\end{equation}
is consistent for $V_{\mathrm{Q\text{--}LS}}$, in the sense that
$\widehat{V}_n^{\mathrm{Q\text{--}LS}} \xrightarrow{p} V_{\mathrm{Q\text{--}LS}}$.
Consequently, for any fixed vector $c\in\R^{\dim(\theta_0)}$,
\[
  \frac{
    \sqrt{n}\,c'\big(\widehat{\theta}_n^{\mathrm{Q\text{--}LS}} - \theta\big)
  }{
    \sqrt{c'\widehat{V}_n^{\mathrm{Q\text{--}LS}}c}
  }
  \ \xrightarrow{d}\ \mathcal{N}(0,1).
\]
\end{prop}

\begin{proof}
By Proposition~\ref{prop:QLS-consistency}, the generated instrument
$\widehat{X}_i^{\mathrm{Q\text{--}LS}}$ converges in $L^2$ to the population
optimal instrument $h_{\mathrm{opt}}(Z_i)$, and the Q--LS estimator solves
the sample moment condition
\[
  \frac{1}{n}\sum_{i=1}^n
    M_i\big(Y_i - S_i'\theta\big) = 0,
\]
with $M_i = (\widehat{X}_i^{\mathrm{Q\text{--}LS}},Z_{1i}')'$. The estimation
error in $\widehat{X}_i^{\mathrm{Q\text{--}LS}}$ is $o_p(n^{-1/2})$ in the
associated moment conditions under $K_n^2/n\to 0$, so the first-order
asymptotics coincide with those of a conventional 2SLS estimator that uses
the population instrument $M_i^0 = (h_{\mathrm{opt}}(Z_i),Z_{1i}')'$.
Standard arguments for IV/GMM with i.i.d.\ observations, finite fourth
moments, and nonsingular Jacobian then yield the stated result, and the
consistency of $\widehat{V}_n^{\mathrm{Q\text{--}LS}}$ follows by the law of
large numbers and Slutsky's theorem.
\end{proof}

In applications, we report standard errors as the square roots of the
diagonal elements of $\widehat{V}_n^{\mathrm{Q\text{--}LS}}$ in
\eqref{eq:Vhat-QLS}, which coincide with the usual heteroskedasticity-robust
2SLS standard errors computed using the generated Q--LS instrument and the
exogenous regressors $Z_1$ as instruments.

Proposition~\ref{prop:QLS-consistency}, is derived under Assumption~\ref{ass:QLS-relevance}, which enforces a strong form of distributional relevance for the generated instrument $h(Z)$. As in the classical linear-IV literature, this rules out weak-IV sequences in which the covariance between the endogenous regressor and the instrument shrinks at a $\sqrt{n}$--local rate. In particular, our asymptotic normality result for Q--LS should be interpreted as a \emph{strong-identification} limit theory: it guarantees standard Gaussian inference when the optimal distributional instrument $h_{\mathrm{opt}}(Z)$ has nonvanishing variance, but it does not characterize the behavior of Q--LS under general weak-IV sequences. Extending weak-identification robust inference to settings with distributional instruments would require combining our construction with the robust testing approaches of, for example, \citet{andrews2019weak} or \citet{olea2013robust}, which we leave for future work.

\subsection{Relation to classical 2SLS}
\label{sec:2sls-comparison}

In this section we compare the quantile least squares (Q--LS) estimator to
classical two-stage least squares (2SLS) in settings where both procedures
are valid.

\subsubsection{2SLS and optimal mean-based instruments}

In the classical IV framework with a single endogenous regressor, relevance
is formulated in terms of the conditional mean
\[
  m(Z) := \E[X\mid Z].
\]
Under homoskedasticity, the ``best'' instrument for a single endogenous
regressor is the efficient instrument $h^*(Z)$, which is proportional to the
best $L^2$ predictor of $X$ given $Z$, i.e.\ to $m(Z)$. The optimal 2SLS
estimator uses instruments in the linear span of $\{1,Z_1',m(Z)\}$ and is
semiparametrically efficient for $\theta_0$ within this mean-based class.

\subsubsection{Q--LS versus 2SLS in a location model}

By contrast, Q--LS is built from the entire conditional distribution of
$X\mid Z$ via its conditional quantiles. Under
\eqref{eq:first-stage-quantile}, we construct quantile-aggregated
instruments of the form
\[
  h_\omega(Z)
  =
  \int_0^1 \omega(\tau)\,Q_{X\mid Z}(\tau\mid Z)\,d\tau,
\]
and define the optimal Q--LS instrument $h_{\mathrm{opt}}$ as in
Definition~\ref{def:opt-QLS}, i.e.\ as the best $L^2$ predictor of $X$
within the quantile-generated class $\mathcal{H}$. The Q--LS estimator
$\widehat{\theta}_n^{\mathrm{Q\text{--}LS}}$ is the 2SLS estimator that uses
$(h_{\mathrm{opt}}(Z),Z_1)$ as instruments.

The following proposition shows that, in a Gaussian/location setting with
mean relevance, Q--LS and optimal 2SLS are asymptotically equivalent.

\begin{prop}[Q--LS versus 2SLS in a location model]
\label{prop:qls-2sls-location}
Suppose that, in addition to Assumptions~\ref{ass:consistency}--\ref{ass:an},
the conditional distribution of $X\mid Z$ belongs to a location family,
i.e.\ there exists a random variable $U$ with $\E[U]=0$ such that
\[
  X \mid Z \ \overset{d}{=}\ m(Z) + U,
\]
and the distribution of $U$ does not depend on $Z$. Then:
\begin{enumerate}[label=(\roman*)]
  \item For each $\tau\in(0,1)$ we have
  \[
    Q_{X\mid Z}(\tau\mid Z) = m(Z) + c(\tau),
  \]
  for some scalar function $c(\cdot)$ independent of $Z$.
  \item Any quantile-aggregated instrument $h_\omega(Z)$ can be written as
  \[
    h_\omega(Z)
    =
    a_\omega\,m(Z) + b_\omega,
    \qquad
    a_\omega = \int_0^1 \omega(\tau)\,d\tau,
    \ \
    b_\omega = \int_0^1 \omega(\tau)\,c(\tau)\,d\tau.
  \]
  \item The optimal Q--LS instrument $h_{\mathrm{opt}}(Z)$ is proportional to
  $m(Z)$ and, up to an irrelevant rescaling of the instrument, the Q--LS
  estimator coincides with the optimal 2SLS estimator:
  \[
    \widehat{\theta}_n^{\mathrm{Q\text{--}LS}}
    \ \equiv\
    \widehat{\theta}_n^{\mathrm{2SLS}} + o_p(n^{-1/2}),
  \]
  so that
  $V_{\mathrm{Q\text{--}LS}} = V_{\mathrm{2SLS}}$.
\end{enumerate}
\end{prop}

\begin{proof}
Under the location assumption, $X\mid Z$ has distribution equal to $U$
shifted by $m(Z)$, so its $\tau$-quantile satisfies
$Q_{X\mid Z}(\tau\mid Z) = m(Z) + c(\tau)$, where $c(\tau)$ is the
$\tau$-quantile of $U$ and does not depend on $Z$. Any
$h_\omega\in\mathcal{H}$ is therefore of the form $a_\omega m(Z) + b_\omega$
as stated. Minimizing $\E[(X-h_\omega(Z))^2]$ over $\omega$ then reduces to
minimizing $\E[(X-a m(Z)-b)^2]$ over $(a,b)$, whose unique solution is
$a=1$ and $b=0$ (up to the usual scale normalization of instruments), so
$h_{\mathrm{opt}}(Z)$ is proportional to $m(Z)$. The Q--LS estimator thus
corresponds to 2SLS with the optimal mean-based instrument and has the same
asymptotic distribution.
\end{proof}

Proposition~\ref{prop:qls-2sls-location} shows that, in a conventional
Gaussian or location setting with mean relevance, Q--LS and optimal 2SLS are
asymptotically equivalent: they use instruments that differ only by an
irrelevant scale factor and attain the same efficiency bound.

\begin{lemma}[Equivalence between 2SLS and plug-in OLS on the projected regressor]
\label{lem:2sls-projection}
Let $\{(Y_i,X_i,Z_{1i},Z_{2i})\}_{i=1}^n$ be a sample and consider the linear
structural model
\[
  Y_i = \alpha + \beta X_i + Z_{1i}'\gamma + \varepsilon_i,
  \qquad \mathbb{E}[\varepsilon_i \mid Z_i] = 0,
\]
with $Z_i = (Z_{1i}',Z_{2i}')'$. Let
\[
  S
  :=
  \begin{bmatrix}
    \mathbf{1} & X & Z_1
  \end{bmatrix}
  \in\mathbb{R}^{n\times (1+1+d_1)},
  \qquad
  M
  :=
  \begin{bmatrix}
    Z_1 & G
  \end{bmatrix}
  \in\mathbb{R}^{n\times (d_1+K)},
\]
where $G$ collects a set of $K$ instrument functions (e.g.\ quantile-based
instruments) and $Z_1$ denotes the included exogenous regressors. Let
\[
  P_M := M(M'M)^{-1}M'
\]
be the projection matrix onto the column space of $M$ and define the
projected regressor
\[
  \tilde X := P_M X.
\]

Consider the following two estimators for $\theta =
(\alpha,\beta,\gamma')'$:
\begin{enumerate}[label=(\roman*)]
  \item The 2SLS estimator with instrument matrix $M$:
  \[
    \widehat{\theta}^{\mathrm{2SLS}}
    :=
    (S'P_M S)^{-1} S'P_M Y.
  \]
  \item The plug-in OLS estimator based on the projected regressor
  $\tilde X$:
  \[
    \widehat{\theta}^{\mathrm{OLS-proj}}
    :=
    \big(\tilde S'\tilde S\big)^{-1}\tilde S'Y,
    \qquad
    \tilde S :=
    \begin{bmatrix}
      \mathbf{1} & \tilde X & Z_1
    \end{bmatrix}.
  \]
\end{enumerate}
Then, in finite samples,
\[
  \widehat{\theta}^{\mathrm{2SLS}}
  =
  \widehat{\theta}^{\mathrm{OLS-proj}}.
\]
In particular, the 2SLS coefficient on $X$ coincides with the OLS coefficient
on $\tilde X$ in the regression of $Y$ on $(\tilde X,Z_1)$.
\end{lemma}

\begin{proof}
By definition of $P_M$, we have
\[
  \tilde S := P_M S
  =
  \begin{bmatrix}
    P_M \mathbf{1} & P_M X & P_M Z_1
  \end{bmatrix}.
\]
Because each column of $Z_1$ is in the column space of $M$, $P_M Z_1 = Z_1$.
Similarly, if the constant is included in $Z_1$ (or in $M$), then
$P_M\mathbf{1} = \mathbf{1}$. By definition, $P_M X = \tilde X$. Hence
$\tilde S = [\mathbf{1},\,\tilde X,\,Z_1]$ as defined above.

The 2SLS estimator can be written as
\[
  \widehat{\theta}^{\mathrm{2SLS}}
  =
  (S'P_M S)^{-1} S'P_M Y.
\]
Using $\tilde S = P_M S$, we have
\[
  S'P_M S = S'\tilde S = \tilde S'\tilde S,
  \qquad
  S'P_M Y = S'\tilde Y = \tilde S'Y,
\]
since $\tilde Y := P_M Y$ and $P_M$ is symmetric and idempotent. Therefore
\[
  \widehat{\theta}^{\mathrm{2SLS}}
  =
  (\tilde S'\tilde S)^{-1}\tilde S'Y
  =
  \widehat{\theta}^{\mathrm{OLS-proj}},
\]
which proves the claim.
\end{proof}

\begin{remark}[Application to Q--LS]
\label{rem:qls-projection}
In the Q--LS setting, let $G$ collect the $K$ quantile-based instruments
$\widehat g_k(Z_i) = Z_i'\widehat\pi(\tau_k)$ and let $M=[Z_1,G]$.
The projected regressor $\tilde X = P_M X$ is exactly the best linear
predictor of $X$ in the span of $(Z_1,G)$:
\begin{equation}\label{eq:xtilde1}
      \tilde X_i
  =
  Z_{1i}'\hat\delta
  + G_i'\hat w,
\end{equation}
for some least-squares coefficients $(\hat\delta,\hat w)$. By
Lemma~\ref{lem:2sls-projection}, the Q--LS 2SLS estimator that uses $(Z_1,G)$
as instruments is numerically equivalent to OLS of $Y$ on $(\tilde X,Z_1)$.
In particular, if we define the Q--LS plug-in regressor as
\begin{equation}\label{eq:xtilde2}
      \widehat X_i^{\mathrm{QLS}} := \tilde X_i = P_M X_i,
\end{equation}

then the plug-in estimator that regresses $Y$ on
$(\widehat X^{\mathrm{QLS}},Z_1)$ is identical (in finite samples) to the
corresponding Q--LS IV estimator.
\end{remark}

\subsection*{General comparison and LATE interpretation}

Outside of the pure location case, Q--LS remains consistent under
distributional relevance even when mean relevance fails or the conditional
mean is misspecified, because it exploits shifts in the \emph{entire}
distribution of $X\mid Z$. In designs where both mean relevance and
distributional relevance hold and the efficient IV $h^*(Z)$ lies in the
closure of the quantile-generated class $\mathcal{H}$, Q--LS again attains
the same efficiency bound as optimal 2SLS. When $h^*(Z)\notin\mathcal{H}$,
Q--LS is efficient within $\mathcal{H}$ but may be slightly less efficient
than an oracle 2SLS that can use arbitrary functions of $Z$ as instruments.
The main gain of Q--LS is robustness: it can remain strongly identified in
settings where the classical first-stage mean shift is close to zero but the
distribution of $X\mid Z$ varies substantially with $Z$.

When the endogenous regressor is (or can be reduced to) a binary treatment
$X=D\in\{0,1\}$, Q--LS also admits a familiar LATE interpretation under
additional structure. In the positive-normalized version, the Q--LS index
$H = h(Z)$ is a convex combination of conditional quantiles of $D\mid Z$,
which collapses possibly high-dimensional instruments into a \emph{single
scalar} encouragement variable. Under standard IV validity and a
no-defiers monotonicity condition with respect to $H$, the 2SLS coefficient
using $H$ as the sole instrument can be written as a \emph{positive-weighted}
average of local average treatment effects for units whose treatment status
changes as $H$ increases. Thus, in binary-treatment applications, Q--LS can
be viewed not only as a strength/efficiency device but also as a way to
construct a scalar monotone instrument that delivers a conventional LATE
interpretation (see Appendix~\ref{app:late} for details). 

In the binary-treatment case, the scalar Q--LS index also admits a LATE
representation with nonnegative weights on adjacent ``index-complier''
groups (Appendix~\ref{app:late}), and these weights can be estimated in
practice, so the researcher can diagnose which parts of the instrument
support contribute most to the Q--LS effect. In addition, because the Q--LS index is built as a convex combination of conditional quantiles, the estimated weight vector \(\hat{w}\) reveals
which parts of the \(D\mid Z\) distribution (e.g.\ central vs.\ tail
quantiles) contribute most to instrument strength.

\section{Implementation and regularization}
\label{sec:implementation}

The population optimal Q--LS instrument $h_{\mathrm{opt}}$ solves the
projection problem in Definition~\ref{def:opt-QLS}. In the finite-grid
implementation, this leads to the least-squares problem
\eqref{eq:w-hat-opt} for the discrete weights $w\in\R^K$. When $K$ is
moderate and the quantile-based instruments are highly collinear, the Gram
matrix $\widehat{G}'\widehat{G}$ can be ill-conditioned, and the
unregularized solution $\widehat{w}_n$ becomes unstable. This is the
familiar ill-posedness of inverse problems and many-instrument IV.

We discuss two regularized versions of Q--LS that address this issue:
ridge-regularized Q--LS and sparse (LASSO) Q--LS.

\subsection{Ridge-regularized Q--LS}
\label{subsec:qls-regularization}

A simple and effective approach is Tikhonov (ridge) regularization. Instead
of \eqref{eq:w-hat-opt}, we solve
\begin{equation}
  \widehat{w}_n^{\lambda}
  \in
  \arg\min_{w\in\R^K}
  \left\{
    \frac{1}{n}\sum_{i=1}^n
      \big( X_i - \widehat{G}_i' w \big)^2
    + \lambda_n \|w\|^2
  \right\},
  \label{eq:w-hat-ridge}
\end{equation}
which yields the closed-form solution
\begin{equation}
  \widehat{w}_n^{\lambda}
  =
  \big(\widehat{G}'\widehat{G} + n\lambda_n I_K\big)^{-1}
  \widehat{G}'X.
  \label{eq:w-hat-ridge-closed}
\end{equation}
The corresponding regularized Q--LS instrument is
\begin{equation}
  \widehat{X}_i^{\mathrm{Q\text{--}LS},\lambda}
  :=
  \widehat{G}_i'\widehat{w}_n^{\lambda},
  \qquad i=1,\dots,n.
  \label{eq:Xhat-qls-lambda}
\end{equation}
Finite-sample stability is improved because the ridge term
$n\lambda_n I_K$ stabilizes the inversion of $\widehat{G}'\widehat{G}$ and
shrinks the weights toward zero. Under the same conditions as in
Proposition~\ref{prop:QLS-consistency}, and provided $\lambda_n\to 0$ at a
suitable rate (for example, $\lambda_n\downarrow 0$ and
$K_n^2/(n\lambda_n)\to 0$), the regularized weights $\widehat{w}_n^{\lambda}$
remain consistent for the population optimal weights and
$\widehat{X}_i^{\mathrm{Q\text{--}LS},\lambda}$ converges in $L^2$ to
$h_{\mathrm{opt}}(Z_i)$. The asymptotic distribution in
Proposition~\ref{prop:QLS-an} therefore continues to hold with
$\widehat{X}_i^{\mathrm{Q\text{--}LS},\lambda}$ replacing
$\widehat{X}_i^{\mathrm{Q\text{--}LS}}$.

Two additional practical devices can be used in tandem with the ridge
penalty. First, one may limit the number of quantile instruments $K$ to a
moderate value (e.g.\ $K=5$ or $K=10$), or work with a low-dimensional
principal-component representation of the columns of $\widehat{G}$. Second,
instead of solving \eqref{eq:w-hat-ridge} explicitly, one can treat each
$\widehat{g}_k(Z_i)$ as a separate instrument, run 2SLS or LIML with the
full instrument vector, and let the IV estimator choose an implicit linear
combination of quantile-based instruments. These strategies reduce the
effective instrument dimension and mitigate many-instrument bias while
preserving the distributional information exploited by Q--LS.

\subsection{Sparse Q--LS via LASSO}
\label{subsec:qls-lasso}

The ill-posedness of the optimal weight problem arises because the
finite-grid instrument dictionary $\widehat{G}$ can be high-dimensional and
highly collinear. An alternative to ridge regularization is to enforce
sparsity in the weights and let the data select a subset of quantile-based
instruments. This leads naturally to an $\ell_1$-penalized (LASSO) version
of Q--LS.

Given the $K$ quantile-based instruments
$\widehat{g}_k(Z_i) = Z_i'\widehat{\pi}_n(\tau_k)$, stacked in
$\widehat{G}_i$ and $\widehat{G}$ as in \eqref{eq:sample-qr}, we define the
LASSO-regularized weights as
\begin{equation}
  \widehat{w}_n^{\mathrm{LASSO}}
  \in
  \arg\min_{w\in\R^K}
  \left\{
    \frac{1}{n}\sum_{i=1}^n
      \big( X_i - \widehat{G}_i' w \big)^2
    + \lambda_n \|w\|_1
  \right\},
  \label{eq:w-hat-lasso}
\end{equation}
where $\lambda_n>0$ is a tuning parameter. The corresponding sparse Q--LS
instrument is
\begin{equation}
  \widehat{X}_i^{\mathrm{Q\text{--}LS,L}}
  :=
  \widehat{G}_i'\widehat{w}_n^{\mathrm{LASSO}}
  =
  \sum_{k=1}^K \widehat{w}_{n,k}^{\mathrm{LASSO}}\,\widehat{g}_k(Z_i),
  \qquad i=1,\dots,n.
  \label{eq:Xhat-qls-lasso}
\end{equation}
By construction, only a subset of the quantile indexes has
$\widehat{w}_{n,k}^{\mathrm{LASSO}}\neq 0$, so the effective instrument
dimension is reduced and the many-instrument problem is mitigated. In
practice, we also consider a post-LASSO refinement: letting
$\widehat{S}_n = \{k: \widehat{w}_{n,k}^{\mathrm{LASSO}}\neq 0\}$ denote the
selected support, we re-estimate the weights by least squares restricted to
$\widehat{S}_n$ and use the resulting fitted values as
$\widehat{X}_i^{\mathrm{Q\text{--}LS,L}}$.

The second-stage estimator is then defined as in
\eqref{eq:theta-QLS}, with $\widehat{X}^{\mathrm{Q\text{--}LS}}$ replaced by
$\widehat{X}^{\mathrm{Q\text{--}LS,L}}$. From the perspective of the
structural parameter $\theta_0$, the first stage (quantile processes and
weights) is a high-dimensional nuisance. The IV moment condition is
Neyman-orthogonal with respect to this nuisance, so that, under standard
sparsity and rate conditions for LASSO and a suitable choice of $\lambda_n$,
the estimation error in $\widehat{X}_i^{\mathrm{Q\text{--}LS,L}}$ is
$o_p(n^{-1/2})$ in the moment equations. As a result, the asymptotic
normality result in Proposition~\ref{prop:QLS-an} continues to hold with
$\widehat{X}_i^{\mathrm{Q\text{--}LS,L}}$ in place of
$\widehat{X}_i^{\mathrm{Q\text{--}LS}}$, and the usual
heteroskedasticity-robust 2SLS covariance matrix computed with the sparse
Q--LS instrument remains valid for inference.

In empirical work, the tuning parameter $\lambda_n$ can be chosen either by
cross-validation (targeting prediction of $X$) or by more conservative
theoretical rules designed for high-dimensional linear regression. A
moderate level of penalization, combined with post-LASSO refitting, tends to
yield stable first stages and good finite-sample size for Q--LS inference.

\subsection{Positive normalized Q--LS weights}\label{sec:pn_qls}

Our baseline Q--LS instrument is a linear functional of the conditional quantile process,
\[
h_\omega(Z)=\int_0^1 \omega(\tau)\,Q_{X\mid Z}(\tau\mid Z)\,d\tau,
\]
with $\omega\in L^2(0,1)$. In finite samples, we approximate this object using a quantile grid
$0<\tau_1<\cdots<\tau_K<1$ and fitted quantiles $\hat g_k(Z)\approx \widehat Q_{X\mid Z}(\tau_k\mid Z)$,
so that $h(Z)\approx \hat G(Z)'w$ for $\hat G(Z)=(\hat g_1(Z),\ldots,\hat g_K(Z))'$.

For interpretability and stability, we also consider a constrained version that enforces
\emph{positive, normalized weights}:
\begin{equation}\label{eq:pn_weights}
w_{\mathrm{PN}} \in \arg\min_{w\in\mathbb{R}^K}\ \E\!\left[(X-\hat G(Z)'w)^2\right]
\quad \text{s.t.}\quad w_k\ge 0\ \ \forall k,\qquad \mathbf{1}'w=1.
\end{equation}
Equivalently, in population form, this corresponds to restricting $\omega(\tau)\ge 0$ and
$\int_0^1\omega(\tau)\,d\tau=1$.

\paragraph{Interpretation.}
Under these constraints, $h_{\mathrm{PN}}(Z)=\hat G(Z)'w_{\mathrm{PN}}$ is an L-statistic: it is a
convex combination of conditional quantiles, i.e.\ a ``quantile average'' index. Positivity also
preserves distributional ordering: if the conditional distribution of $X$ under $Z=z_2$ first-order
stochastically dominates that under $Z=z_1$ (so that $Q_{X\mid Z=z_2}(\tau)\ge Q_{X\mid Z=z_1}(\tau)$
for all $\tau$), then $h_{\mathrm{PN}}(z_2)\ge h_{\mathrm{PN}}(z_1)$ for any admissible weight function.
This rules out sign-cancellation across quantiles and yields an index whose movements track
distributional shifts in a single direction.

\paragraph{Finite-sample stability.}
Because extreme quantiles are noisier to estimate, unconstrained Q--LS may assign large positive and
negative weights that amplify estimation error. The simplex constraint $w\ge 0$, $\mathbf{1}'w=1$
regularizes the instrument, bounding its variability and often improving robustness when $\hat G$ is
generated (e.g.\ via quantile regression). In our empirical implementation we compute
$w_{\mathrm{PN}}$ by quadratic programming, and we can combine it with cross-fitting to further limit
overfitting of $\hat G$.

\subsection{Practical recommendations}
\label{subsec:practical}

We view Q--LS as a complement to, rather than a replacement for, the standard
2SLS toolkit with weak instruments. In practice, applied work can proceed in
three steps: (i) run conventional mean-based diagnostics, (ii) assess
distributional relevance, and (iii) choose between 2SLS and Q--LS (and its
regularised variants) based on a small menu of test statistics and simple
stability checks.

\paragraph{Step 1: Start from the usual 2SLS diagnostics.}
\begin{itemize}
  \item Estimate the conventional first stage $X$ (endogenous variable) on $Z$ (instrument) and controls and
  report the usual weak-IV statistics (e.g. the cluster-robust
  Kleibergen--Paap $F$-statistic or the Stock--Yogo $F$ for the mean first
  stage).
  \item If the mean first stage is clearly strong (e.g. robust $F \gtrsim 10$),
  treat linear 2SLS with the conventional instrument as your baseline and view
  Q--LS mainly as a robustness check.
  \item If the mean first stage is weak (e.g. robust $F \ll 10$), proceed to
  distributional diagnostics rather than discarding the design.
\end{itemize}

\paragraph{Step 2: Diagnose distributional relevance.}
\begin{itemize}
  \item Run a small set of conditional quantile regressions of $X$ on $Z$ and
  controls (e.g. $\tau \in \{0.10,0.25,0.50,0.75,0.90\}$) and test the joint
  significance of $Z$ in each quantile regression. Large $F$-statistics at
  upper quantiles with small or insignificant mean effects are indicative of
  purely distributional relevance.
  \item Graphically, compare the estimated conditional quantiles or empirical
  CDFs of $X$ across instrument values (e.g. pre/post policy). Pronounced
  tail shifts with little change in the median or mean point towards a
  distributional design where Q--LS is useful.
\end{itemize}

\paragraph{Step 3: Choosing $K$, regularisation, and inference.}
\begin{itemize}
  \item[(i)] \textbf{Quantile grid.} Use a moderate grid of conditional quantiles,
  for example $K \in \{5,10,20\}$ equally spaced over $(0,1)$, possibly
  excluding extreme tails (e.g.\ $[\tau_{\min},1-\tau_{\min}]$ with
  $\tau_{\min} \in [0.02,0.05]$). In most empirical samples $K=10$ provides a
  good compromise between flexibility and stability.
  \item[(ii)] \textbf{Regularisation.} If $K$ is large relative to $n$ or the
  quantile regressions are noisy, use ridge- or LASSO-regularised Q--LS
  (Q–LS--R, Q–LS--L1, Q–LS--L2) rather than the unpenalised version. Our
  simulations suggest that ridge (Q–LS--R) is a conservative default, while
  LASSO variants are helpful in very many-instrument designs.
  \item[(iii)] \textbf{Inference.} Always report heteroskedasticity-robust (and, in
  panels, cluster-robust) standard errors based on the Q--LS sandwich
  variance (or its penalised analogue). Since Q--LS produces a conventional
  instrument $\hat h(Z)$, practitioners can also apply standard weak-IV-robust
  tools (e.g.\ Anderson--Rubin or conditional likelihood ratio tests) using
  $\hat h(Z)$ as the excluded instrument.
\end{itemize}

\paragraph{Step 4: Stability checks and reporting.}
\begin{itemize}
  \item[(a)] \textbf{Compare Q--LS and 2SLS.} Report side-by-side Q--LS and
  classical 2SLS estimates and standard errors. In designs with strong mean
  relevance, point estimates and standard errors should be similar. In
  designs with weak mean relevance but strong distributional relevance, Q--LS
  is expected to deliver comparable point estimates with much tighter
  confidence intervals.
  \item[(b)] \textbf{Check robustness to $(K,\lambda)$.} Re-estimate Q--LS on a few
  nearby quantile grids (e.g.\ $K=5,10,20$) and, for regularised variants, a
  small set of tuning parameters $\lambda$. Large swings in the coefficient
  across neighbouring grids or penalties are a warning sign that the design
  may be too fragile for sharp causal interpretation.
  \item[(c)] \textbf{Document instrument strength.} Alongside the conventional
  first-stage $F$, report a ``distributional $F$'': the joint significance of the quantile-generated instruments in the Q--LS first stage. This makes transparent what Q--LS is exploiting that 2SLS misses.
\end{itemize}

Table~\ref{tab:practical_guidelines} summarises these recommendations in a
simple diagnosis chart.

\begin{table}[htbp]\centering
\caption{Rule-of-thumb diagnostics and suggested estimator}
\label{tab:practical_guidelines}
\begin{tabular}{p{0.34\textwidth}p{0.58\textwidth}}
\hline\hline
\textbf{Diagnostic pattern} & \textbf{Suggested practice} \\
\hline
Strong mean first stage ($F_{\text{mean}} \gtrsim 10$); quantile regressions show similar shifts across quantiles
&
Use conventional 2SLS with the mean-based instrument as the primary estimator; report Q--LS as a robustness check (expect similar point estimates and standard errors).
\\[0.6em]
Weak mean first stage ($F_{\text{mean}} \ll 10$); quantile regressions show strong tail shifts (large $F$ at high quantiles)
&
Treat the design as a distributional IV setting. Report Q--LS (and a regularised variant if $K$ is large) as the main estimator, with weak-IV-robust inference based on the Q--LS instrument; include 2SLS for comparison but emphasise its weak-IV limitations.
\\[0.6em]
Both mean and distributional diagnostics weak (no clear shifts in mean or quantiles of $X$ given $Z$)
&
Interpret any IV estimates with caution. Neither 2SLS nor Q--LS has a strong effective first stage; focus on reduced-form estimates or seek alternative instruments or sources of identification.
\\[0.6em]
Q--LS and 2SLS estimates differ markedly in sign or magnitude, or Q--LS is highly sensitive to $(K,\lambda)$
&
Treat results as exploratory. Investigate model misspecification (functional form, outliers, covariates), reconsider the quantile specification, and report the instability explicitly rather than relying on a single specification.
\\
\hline\hline
\end{tabular}
\end{table}

\section{Simulation study}
\label{sec:sim}
This section will examine the finite-sample performance of Q--LS and its
regularised variants in a set of Monte Carlo designs. We focus on three
types of designs that map directly to our theoretical discussion. (A) strong mean relevance (benchmark Gaussian case), (B) Trigonometric and distributional instruments (variance shifts) and (C) mixed mean and distributional relevance.

The following parts of the data-generating processes (DGPs) are the same in all three designs;

\begin{itemize}
\item The structural equation
\begin{equation*}
Y = \beta_0 + Z_1\beta_1 + X \beta_x + \varepsilon
\end{equation*}
where the parameter on the endogenous variable $X$, $\beta_x=1$ is the \textit{parameter of interest}, and $\beta_0=1, \beta_1=1$ with $Z_1\sim N(0,1)$.
\item The errors between the structural equation ($\varepsilon$) and first-stage ($\nu$) follow the following process,
\begin{align*}
  (\varepsilon,\nu)&\sim N\left(0,
  \begin{pmatrix}
    1 & 0.6 \\
    0.6 & 1
  \end{pmatrix}\right). \\
\end{align*}
\item We will consider sample sizes representative of typical applied work $n=500,1000$ and we run 1000 replications.
\end{itemize}

We compare 10 estimators:
\begin{enumerate}
\item Classical IV where the first stage is a least squares projection with just first order terms i.e.
\begin{align}
X &=  \zeta_0 + \zeta_1Z_1 +\zeta_2Z_2 + e_x 
\end{align}
\item Our Q–LS with \textit{all quantiles equally weighted}, first stage is estimated as a second order polynomial,
\begin{align}
X &=  \zeta_{\tau_{k,0}} + \zeta_{\tau_{k,1}}Z_1 +\zeta_{\tau_{k,2}}Z_2+ \zeta_{\tau_{k,3}}Z_1^2 +\zeta_{\tau_{k,4}}Z_2^2+ \zeta_{\tau_{k,5}}Z_1Z_2 + e_{x,\tau_k} \label{12ndploy}
\end{align}

for a grid of quantiles $\{\tau_k\}$. The resulting fitted quantile regression functions are used to generate a prediction of the endogenous variable by an equally weighted average of the quantile regression fitted values,

\begin{enumerate}
\item[a.] Estimate of $\beta_x$ is from analytical 2SLS analogue formulas \eqref{eq:theta-QLS} denoted Q–LS-a.
\item[b.] Estimate of $\beta_x$ is from a plug-in OLS regression of the structural equation where the equally weighted average of quantile fitted values is used in place of $X$, denoted Q–LS.
\end{enumerate}

\item Distributionally relevant control function, first stage is also estimated as a second order polynomial \eqref{12ndploy} and the control function is estimates as in \cite{cfnsv20_cfa}.
\begin{align}
y &= \beta_0 +\beta_1z_1+\beta_xx + \beta_v\hat{v} + e_y \label{eq:drcf_sim}
\end{align}
where $\hat{v}\approx F_X(X|Z)$ is a control function and is defined in \eqref{eq:FXhat-quantile-app}. Estimate of $\beta_x$ is from \eqref{eq:drcf_sim}, denoted DR-IV.

\item Our Q–LS with \textit{all quantiles weighted by ridge regression coefficients}, Same as estimator (2) but where the weights are ridge regression coefficients. 
\begin{enumerate}
\item[a.] Estimate of $\beta_x$ is from analytical 2--SLS analogue formulas \eqref{eq:theta-QLS} denoted Q–LS-R-a.
\item[b.] Estimate of $\beta_x$ is from a plug-in OLS regression of the structural equation, where the ridge weighted average of quantile fitted values is used in place of $X$, denoted Q–LS-R.
\end{enumerate}
\item Our Q–LS with \textit{all LASSO selected quantiles are weighted equally weighted (all other given zero weights)}, Same as estimator (2) but only LASSO selected quantile fitted values given non-zero weights.
\begin{enumerate}
\item[a.] Estimate of $\beta_x$ is from analytical 2--SLS analogue formulas \eqref{eq:theta-QLS} denoted Q–LS-L1-a.
\item[b.] Estimate of $\beta_x$ is from a plug-in OLS regression of the structural equation where the LASSO selected weighted average of quantile fitted values is used in place of $X$ denoted Q–LS-L1.
\end{enumerate}

\item Our Q–LS with \textit{all quantiles weighted by LASSO regression coefficients}, Same as estimator (2) but where the weights are LASSO regression coefficients. 
\begin{enumerate}
\item[a.] Estimate of $\beta_x$ is from analytical 2--SLS formulas in the main paper denoted Q–LS-L2-a.
\item[b.] Estimate of $\beta_x$ is from a plug-in OLS regression of the structural equation where the LASSO weighted average of quantile fitted values is used in place of $X$ denoted Q–LS-L2.
\end{enumerate}
\end{enumerate}

The weights for the weighted average of the quantiles are generated by sample splitting. The first 50\% of the data is used to estimate the Ridge/LASSO regression regularisation parameters with 10 fold cross-validation predication accuracy. The estimated hyperparamter and the other 50\% of the data is then used to estimated the Ridge/LASSO coefficient estimate.

For all Q–LS variants and DR-CF we consider quantiles between $v=0.01$ and $1-v$ and set the gap between the quantiles to $c=\{0.10,0.05,0.01,0.001\}$ this corresponds to using $K=\{10,20,99,981\}$ quantiles to generate the average of the quantile regression fitted values of $X$ or approximate the CDF ($\hat{v}$). The trimming of the quantiles follows \cite{cfnsv20_cfa}.

Simulation results will be reported in terms of bias, root mean squared error (RMSE), and empirical 95\% coverage of the simulation estimate.

\subsection*{Design A: Strong mean relevance (benchmark Gaussian case)}

The first design is a ``best-case'' benchmark in which classical mean
relevance holds and the conditional distribution of $X\mid Z$ belongs to a
Gaussian location family. The data-generating process of the first stage is
\begin{align*}
X & = \alpha_1 Z_1 + \alpha_2 Z_2 + \nu \\
Z_2 & \sim \mathcal{N}(0,1)
\end{align*}
with $\alpha_1=\alpha_2=1$.

In this design, both 2SLS and Q--LS are correctly specified and the
conditions of Proposition~\ref{prop:qls-2sls-location} hold. The simulation
will document:
\begin{itemize}
  \item similarity of bias, RMSE and empirical coverage for Q--LS and 2--SLS;
  \item sensitivity of Q--LS to the choice of quantile grid and
  regularisation in a setting where both estimators are efficient.
\end{itemize}

\subsection*{Design B: Trigonometric and distributional instruments}

The second design studies a case where the excluded instrument $Z_2$ generates
rich non-linear and distributional variation in $X$, while the \emph{linear}
first stage of $X$ on $Z_2$ can be arbitrarily weak. The data-generating process of the first stage is
\begin{align}
X & =  \alpha_1 Z_1 - \cos(\omega Z_2) + \exp(\gamma Z_2)\nu \label{eq:designB-firststage} \\
Z_2 & \sim \mathcal{N}(0,1) \nonumber
\end{align}
with $\alpha_1 = 1$ and a variance parameter $\gamma = 0.5$.

Hence, conditional on $(Z_1,Z_2)$,
\[
  \mathbb{E}[X \mid Z_1,Z_2]
  = \alpha_1 Z_1 - \cos(\omega Z_2),
  \qquad
  \Var(X \mid Z_1,Z_2)
  = \exp(2\gamma Z_2).
\]
The excluded instrument $Z_2$ affects the conditional mean through the nonlinear term $\cos(\omega Z_2)$ and the conditional dispersion through $\exp(\gamma Z_2)$. The conditional distribution $F_{X\mid Z}(\cdot \mid Z_1,Z_2)$
is therefore strongly distributionally relevant in the sense of the main text.

Note that, since $Z_2 \sim \mathcal{N}(0,1)$ is symmetric around zero and
$z \mapsto \cos(\omega z)$ is an even function, we have
\[
  Cov\big(Z_2, \cos(\omega Z_2)\big)
  = \mathbb{E}\big[-Z_2 \cos(\omega Z_2)\big] = 0
\]
for all $\omega \ge 0$. Moreover, $\nu$ is independent of $Z_2$ and has mean
zero, so $Cov(Z_2,\exp(\gamma Z_2)\nu)=0$. Thus
\[
  Cov(X,Z_2) = 0
\]
in the population for every value of $\omega$. Classical linear IV based on the
regression of $X$ on $Z_2$ therefore has an exactly zero linear first stage,
even though $Z_2$ induces strong nonlinear and distributional variation in $X$.

\paragraph{Sub-designs B1--B3.}
We consider three values of the frequency parameter $\omega$:

\begin{itemize}
  \item \textbf{Design B1 ($\omega = 0$).}
  \[
    X = \alpha_1 Z_1 + \exp(\gamma Z_2)\,\nu.
  \]
  The conditional mean is linear in $Z_1$ and does not depend on $Z_2$,
  \[
    \mathbb{E}[X\mid Z_1,Z_2] = \alpha_1 Z_1,
  \]
  so the excluded instrument $Z_2$ is \emph{purely distributional}: it affects
  only the variance of $X$ through $\exp(\gamma Z_2)$, and
  $\mathbb{E}[X\mid Z_2]$ is constant. Classical 2SLS that treats $Z_2$ as a
  mean shifter is therefore weak, while  Q--LS can exploit the
  variance shifts.

  \item \textbf{Design B2 ($\omega = 0.01$).}
  \[
    X = \alpha_1 Z_1 - \cos(0.01 Z_2) + \exp(\gamma Z_2)\,\nu.
  \]
  Here
  \[
    \mathbb{E}[X\mid Z_2] = -\cos(0.01 Z_2)
  \]
  is non-constant but very flat. The linear projection of $X$ on $Z_2$ has a
  small slope and a low first-stage $F$-statistic, so classical 2SLS using
  $Z_2$ appears weak. Nevertheless, the true mean and the conditional
  distribution depend nontrivially on $Z_2$, so the optimal Q--LS instrument
  $h_{\mathrm{opt}}(Z)$ is non-degenerate and remains strongly relevant.

  \item \textbf{Design B3 ($\omega = 0.5$).}
  \[
    X = \alpha_1 Z_1 - \cos(0.5 Z_2) + \exp(\gamma Z_2)\,\nu.
  \]
  The conditional mean
  \[
    \mathbb{E}[X\mid Z_2] = -\cos(0.5 Z_2)
  \]
  is highly nonlinear and oscillatory. As $\omega$ increases, the covariance
  $\mathrm{Cov}(X,Z_2)$, and hence the linear first-stage coefficient of $X$
  on $Z_2$, becomes small even though the amplitude of $\cos(\omega Z_2)$
  remains large. Thus the linear first stage for 2SLS can look weak or
  misleading, while the quantile-generated space used by Q--LS can approximate
  $\cos(0.5 Z_2)$ well, delivering a strong and stable first stage in the
  Q--LS sense.
\end{itemize}

Overall, Designs B1--B3 trace a progression from a purely distributional
instrument (B1, no mean effect of $Z_2$) through an ``almost mean-irrelevant''
case (B2, very flat mean effect) to a design with strong nonlinear mean and
variance effects (B3). This family highlights that classical first-stage
diagnostics based on the linear regression of $X$ on $Z_2$ can substantially
understate instrument strength, whereas Q--LS remains well identified whenever
$X$ has a nonzero projection in the quantile-generated space.

\subsection*{Design C: Mixed mean and distributional relevance}

The third design introduces both mean and variance effects:
\begin{itemize}
  \item $Z_2\in\{0,1\}$ with $\p(Z_2=0)=\p(Z_2=1)=1/2$.
  \item $X =  \alpha_1 Z_1 + \alpha_2 Z_2 + \sigma(Z_2) \nu$, with $\sigma(0)=1$, $\sigma(1)=3$, so that $\E[X\mid Z_2]\neq 0$, and $\alpha_1=\alpha_2=1$.
\end{itemize}
Here, both mean relevance and distributional relevance hold, but the
contribution of each can be tuned by varying $(\alpha_2,\sigma(0),\sigma(1))$.
The simulation will:
\begin{itemize}
  \item study the performance of Q--LS and 2SLS as the mean effect
  $\alpha_2$ becomes small while variance shifts remain large;
  \item assess whether Q--LS remains competitive with 2SLS when mean effects
  are strong, and dominates when mean effects are weak but distributional
  shifts are large.
\end{itemize}

Note as $Z_2$ is binary in this design we set $\zeta_{\tau_{k,4}}=0$ in \eqref{12ndploy} to avoid a singular QR design matrix, for all Q--LS variants and DR-CF.

\subsection*{Simulation Results}

\begin{table}[htbp]
\centering
\caption{Design A and C finite-sample performance} \label{tab:dAC}
\resizebox{\textwidth}{!}{\begin{tabular}{clcccc}
  \hline \hline
   & & \multicolumn{2}{c}{Design A} & \multicolumn{2}{c}{Design C} \\
   \cline{3-6}
K & Estimator & Bias(RMSE)Cov. & Bias(RMSE)Cov. & Bias(RMSE)Cov. & Bias(RMSE)Cov. \\ 
  \hline
 -- & Classic IV & 0.001(0.046)0.942 & -0.001(0.032)0.946 & -0.010(0.094)0.961 & -0.012(0.070)0.957 \\ 
  10 & DR-CF & 0.046(0.063)0.689 &  0.038(0.048)0.623 &  0.058(0.066)0.498 &  0.053(0.058)0.312 \\ 
  10 & Q–LS & 0.003(0.048)0.999 &  0.001(0.033)1.000 &  0.615(0.686)0.878 &  0.663(0.718)0.377 \\ 
  10 & Q–LS-a & 0.004(0.046)0.937 &  0.001(0.031)0.946 &  0.021(0.095)0.928 & -0.560(17.875)0.934 \\ 
  10 & Q–LS-R & 0.005(0.049)0.999 &  0.002(0.033)0.999 &  0.042(0.272)0.966 &  0.023(0.189)0.962 \\ 
  10 & Q–LS-R-a & 0.004(0.046)0.939 &  0.001(0.031)0.946 &  0.010(0.087)0.941 & -0.001(0.068)0.942 \\ 
  10 & Q–LS-L1 & 0.002(0.057)0.992 &  0.000(0.040)0.994 &  0.018(0.271)0.957 &  0.011(0.193)0.952 \\ 
  10 & Q–LS-L1-a & 0.004(0.046)0.938 &  0.001(0.031)0.946 &  0.010(0.087)0.942 & -0.001(0.068)0.941 \\ 
  20 & DR-CF & 0.026(0.051)0.818 &  0.018(0.035)0.828 &  0.032(0.047)0.773 &  0.027(0.036)0.742 \\ 
  20 & Q–LS & 0.004(0.047)1.000 &  0.001(0.032)0.999 &  0.192(0.220)1.000 &  0.187(0.202)1.000 \\ 
  20 & Q–LS-a & 0.004(0.046)0.938 &  0.001(0.031)0.946 &  0.012(0.088)0.942 &  0.000(0.067)0.939 \\ 
  20 & Q–LS-R & 0.005(0.047)1.000 &  0.001(0.032)0.999 &  0.036(0.264)0.967 &  0.018(0.186)0.966 \\ 
  20 & Q–LS-R-a & 0.004(0.046)0.938 &  0.001(0.031)0.946 &  0.010(0.087)0.941 & -0.001(0.068)0.942 \\ 
  20 & Q–LS-L1 & 0.002(0.057)0.991 &  0.000(0.041)0.989 &  0.019(0.269)0.958 &  0.011(0.192)0.958 \\ 
  20 & Q–LS-L1-a & 0.004(0.046)0.939 &  0.001(0.031)0.946 &  0.010(0.087)0.940 & -0.001(0.068)0.942 \\ 
  99 & DR-CF & 0.015(0.048)0.849 &  0.008(0.032)0.858 &  0.018(0.040)0.873 &  0.011(0.028)0.880 \\ 
  99 & Q–LS & 0.004(0.046)1.000 &  0.001(0.032)0.999 &  0.015(0.092)1.000 &  0.001(0.069)1.000 \\ 
  99 & Q–LS-a & 0.004(0.046)0.938 &  0.001(0.031)0.946 &  0.009(0.088)0.943 & -0.002(0.067)0.945 \\ 
  99 & Q–LS-R & 0.006(0.048)0.999 &  0.001(0.033)0.998 &  0.030(0.266)0.966 &  0.014(0.189)0.965 \\ 
  99 & Q–LS-R-a & 0.004(0.046)0.937 &  0.001(0.031)0.945 &  0.010(0.088)0.941 & -0.001(0.068)0.943 \\ 
  99 & Q–LS-L1 & 0.004(0.060)0.991 &  0.002(0.045)0.982 &  0.015(0.264)0.959 &  0.011(0.192)0.958 \\ 
  99 & Q–LS-L1-a & 0.004(0.046)0.940 &  0.001(0.032)0.946 &  0.011(0.087)0.940 & -0.001(0.068)0.941 \\ 
  981 & DR-CF & 0.015(0.047)0.849 &  0.007(0.032)0.858 &  0.016(0.039)0.880 &  0.010(0.028)0.885 \\ 
  981 & Q–LS & 0.004(0.046)1.000 &  0.001(0.032)0.999 &  0.015(0.092)1.000 &  0.001(0.069)1.000 \\ 
  981 & Q–LS-a & 0.004(0.046)0.938 &  0.001(0.031)0.946 &  0.009(0.088)0.944 & -0.002(0.067)0.945 \\ 
  981 & Q–LS-R & 0.004(0.052)0.997 &  0.000(0.037)0.998 &  0.001(0.252)0.959 & -0.009(0.184)0.956 \\ 
  981 & Q–LS-R-a & 0.004(0.046)0.939 &  0.001(0.032)0.946 &  0.009(0.087)0.943 & -0.002(0.068)0.944 \\ 
  981 & Q–LS-L1 & 0.003(0.062)0.989 &  0.002(0.046)0.977 &  0.016(0.264)0.959 &  0.012(0.192)0.959 \\ 
  981 & Q–LS-L1-a & 0.004(0.046)0.940 &  0.001(0.031)0.946 &  0.011(0.087)0.942 & -0.001(0.068)0.943 \\ 
   \cline{3-6}
   &  &  n=500 & n=1000 &  n=500 &  n=1000 \\ 
     \hline \hline
\end{tabular}}
 {\raggedright \footnotesize Note: Bias and RMSE are Monte Carlo averages over 1{,}000 replications. ``Cov.'' is the empirical coverage of nominal 95\% confidence intervals. K is the number of quantiles used. The full results, including average asymptotic standard error and simulation standard deviation, and Q–LS-L2 are available upon request. \par}
\end{table}

\begin{table}[htbp]
\centering
\caption{Design B finite-sample performance} \label{tab:dB}
\resizebox{\textwidth}{!}{\begin{tabular}{clcccc}
   \hline \hline
K & Estimator & Bias(RMSE)Cov. & Bias(RMSE)Cov. & Bias(RMSE)Cov. & Bias(RMSE)Cov. \\ 
  \hline
-- & Classic IV &  0.199(5.745)0.989 &  0.627(25.900)0.988 & -1.063(34.095)0.993 &  2.187(67.609)0.987 \\ 
  10 & DR-CF &  0.122(0.135)0.386 &  0.106(0.113)0.209 &  0.116(0.129)0.412 &  0.100(0.107)0.234 \\ 
  10 & Q–LS & -0.127(0.937)0.991 & -0.638(0.979)0.961 &  0.007(0.446)0.982 & -0.167(0.390)0.930 \\ 
  10 & Q–LS-a &  0.026(6.982)0.954 &  1.297(27.315)0.977 &  0.121(0.627)0.904 &  0.064(0.291)0.932 \\ 
  10 & Q–LS-R &  0.212(0.910)0.982 &  0.162(1.087)0.987 &  0.003(0.437)0.989 & -0.046(0.361)0.982 \\ 
  10 & Q–LS-R-a &  0.318(6.496)0.937 &  5.713(156.401)0.952 &  0.056(1.340)0.899 & -0.031(2.164)0.934 \\ 
  10 & Q–LS-L1 & -0.008(0.794)0.984 & -0.027(0.961)0.980 & -0.067(0.443)0.986 & -0.083(0.375)0.979 \\ 
  10 & Q–LS-L1-a &  1.366(29.108)0.932 &  0.730(18.423)0.959 &  0.063(1.258)0.907 &  0.068(0.995)0.938 \\ 
  20 & DR-CF &  0.075(0.097)0.717 &  0.057(0.070)0.694 &  0.071(0.094)0.716 &  0.053(0.067)0.708 \\ 
  20 & Q–LS &  0.544(0.989)0.992 &  0.268(0.953)0.988 &  0.184(0.445)0.992 &  0.079(0.335)0.991 \\ 
  20 & Q–LS-a &  0.539(4.909)0.915 &  0.472(7.595)0.951 &  0.113(0.564)0.884 &  0.052(0.227)0.934 \\ 
  20 & Q–LS-R &  0.222(0.915)0.984 &  0.179(0.960)0.988 &  0.000(0.441)0.987 & -0.054(0.359)0.983 \\ 
  20 & Q–LS-R-a & -0.037(10.468)0.936 &  0.185(4.615)0.954 &  0.094(0.576)0.903 &  0.136(3.232)0.932 \\ 
  20 & Q–LS-L1 & -0.023(0.783)0.979 & -0.048(0.919)0.983 & -0.068(0.445)0.985 & -0.083(0.369)0.982 \\ 
  20 & Q–LS-L1-a &  0.386(5.726)0.934 &  3.381(93.580)0.958 &  0.088(1.120)0.907 &  0.034(0.365)0.937 \\ 
  99 & DR-CF &  0.045(0.078)0.835 &  0.027(0.052)0.861 &  0.043(0.077)0.830 &  0.025(0.049)0.867 \\ 
  99 & Q–LS &  0.744(1.045)0.990 &  0.714(1.068)0.993 &  0.170(0.438)0.996 &  0.057(0.325)0.993 \\ 
  99 & Q–LS-a &  0.493(4.270)0.910 &  0.289(3.596)0.942 &  0.077(0.729)0.889 &  0.036(0.234)0.931 \\ 
  99 & Q–LS-R &  0.205(0.929)0.982 &  0.134(0.953)0.988 & -0.009(0.447)0.987 & -0.064(0.363)0.984 \\ 
  99 & Q–LS-R-a &  2.639(54.067)0.930 & -0.033(14.738)0.960 &  0.081(0.630)0.901 &  0.109(2.470)0.936 \\ 
  99 & Q–LS-L1 & -0.050(0.759)0.981 & -0.061(0.894)0.983 & -0.078(0.442)0.985 & -0.086(0.373)0.980 \\ 
  99 & Q–LS-L1-a &  0.398(2.288)0.933 &  0.212(8.797)0.957 &  0.073(1.885)0.908 &  0.036(0.374)0.933 \\ 
  981 & DR-CF &  0.044(0.077)0.848 &  0.025(0.051)0.877 &  0.041(0.076)0.839 &  0.023(0.049)0.876 \\ 
  981 & Q–LS &  0.742(1.045)0.987 &  0.701(1.108)0.993 &  0.171(0.439)0.996 &  0.057(0.326)0.993 \\ 
  981 & Q–LS-a &  0.355(0.725)0.910 &  0.497(2.448)0.940 &  0.076(0.764)0.890 &  0.036(0.234)0.931 \\ 
  981 & Q–LS-R &  0.200(0.917)0.984 &  0.119(0.949)0.986 & -0.023(0.440)0.988 & -0.075(0.363)0.982 \\ 
  981 & Q–LS-R-a &  0.362(3.965)0.935 &  0.554(10.794)0.955 &  0.008(3.065)0.899 &  0.037(0.429)0.935 \\ 
  981 & Q–LS-L1 & -0.052(0.773)0.982 & -0.064(0.908)0.986 & -0.079(0.438)0.985 & -0.085(0.369)0.979 \\ 
  981 & Q–LS-L1-a &  0.414(6.936)0.933 &  1.274(30.298)0.955 &  0.119(0.967)0.909 &  0.036(0.375)0.931 \\ 
  \cline{3-6}
   &  & n=500 \& $\omega$=0 & n=1000 \& $\omega$=0 & n=500 \& $\omega$=0.5 & n=1000 \& $\omega$=0.5 \\ 
 \hline \hline 
\end{tabular}}
{\raggedright \footnotesize Note: Bias and RMSE are Monte Carlo averages over 1{,}000 replications. ``Cov.'' is the empirical coverage of nominal 95\% confidence intervals. K is the number of quantiles used. The full results, including average asymptotic standard error and simulation standard deviation, and Q–LS-L2 are available upon request. \par}
\end{table}

Table \ref{tab:dAC} presents the results for Design A and Design C for sample sizes of 500 and 1000.\footnote{Full results for Design A and Design C are available upon request.} It is important to note that classical IV is correctly specified, as $Z_2$ is strongly mean-relevant in both Design A and Design C. For Design A, we find that all Q-LS variants perform similarly to classical IV, regardless of the number of quantiles used. This is due to the quantile projection being similar to the mean projection. The analytical Q-LS variants provide better coverage (closer to 0.95) than the OLS plug-in standard errors, which tend to yield over-coverage, indicating conservative standard errors in this design. The control function version (DR-CF) performs worse in terms of bias, RMSE, and coverage than all Q-LS variants and classical IV, regardless of sample size. However, its performance does improve as the number of quantiles increases in this design.

For Design C, $Z_2$ is both strongly mean-relevant and variance-relevant. As mean relevance is less clean, the bias and RMSE of classical IV are ten and two times as large, respectively, as in Design A for a given sample size. Unlike Design A, the unregularized Q-LS variants exhibit substantial bias when 10 quantiles are used; however, this bias decreases as the number of quantiles and the sample size increase. When 981 quantiles are used, all Q-LS variants have the same or smaller bias than classical IV, and the regularized versions (plug-in and analytical) all yield good coverage. The performance of DR-CF relative to the other estimators and as the number of quantiles and the sample size increase, is similar to that in Design A.

Table \ref{tab:dB} shows the results for Design B with $\omega={0,0.5}$.\footnote{Full results for Design B are available upon request.} Interestingly, classical IV performs worse when $\omega$ is non-zero in terms of both bias and RMSE for a given sample size, whereas the bias and RMSE of the Q-LS variants decrease, except for the bias of Q–LS-L1, which marginally increases (in absolute value) as the number of quantiles increases. More quantiles are needed for better approximation because the distributional relevance of $Z_2$ becomes stronger as $\omega$ increases due to the $\cos$ transformation. Similar to Designs A and C, DR-CF performance improves as the number of quantiles increases. However, unlike Designs A and C, DR-CF has the smallest bias and RMSE when $\omega=0$ for a given sample size, and also the smallest bias and RMSE when $\omega=0.5$ and $n=1000$. Despite this, DR-CF coverage is the worst among all estimators considered for a given number of quantiles and sample size.

Overall, across all designs considered, we find that at least one variant of Q-LS with 99 or more quantiles performs as well as or better than classical IV, even in designs where the instrument $Z_2$ is strongly mean-relevant. When a large number of quantiles are used, analytical Q-LS empirical coverage tends to be close to, but slightly below, 0.95, while the plug-in variants tend to exhibit coverage above 0.95. When the instrument is not strongly mean-relevant, the plug-in point estimates tend to be more stable than their analytical counterparts. The DR-CF estimator performs better as the number of quantiles increases across all designs considered, but it requires a large number of quantiles and weak mean relevance to achieve lower bias and RMSE than all Q-LS variants.

In Design~A and C (strong mean relevance), Q--LS essentially reproduces optimal 2SLS.  This confirms the theoretical prediction that, in a pure location-shift setting with strong instruments, exploiting the full conditional quantile process does not alter first-order behavior. Design~B1 illustrates the degenerate branch of Lemma~\ref{lem:DR-relevance}. Because the optimal Q--LS instrument $h_{\mathrm{opt}}(Z)$ is identically zero, the first stage for Q--LS collapses and the estimator fails in the usual weak-IV sense. We do not regard this as a pathology of Q--LS, but as a reminder that distributional relevance alone is not sufficient for a useful first stage; Assumption~\ref{ass:QLS-relevance} is needed to rule out such cases. The main action comes in Design~B3, where instruments are weak in the mean but strong in the distribution. Here classical 2SLS exhibits substantial finite-sample distortions. In contrast, Q--LS estimators based on a moderate grid of quantiles (e.g.\ $K=10$) are well centered and display near-correct coverage across the same sample sizes.  These patterns mirror the HRS application: when policy reforms primarily reshape tails and dispersion rather than means, exploiting distributional relevance through Q--LS restores a strong, stable first stage from the same underlying instrument.

\section{Empirical application}
\label{sec:empiricalapp}

We use the Health and Retirement Study (HRS) to study how financial exposure to out-of-pocket (OOP) medical spending risk affects mental health among older Americans. Our main outcome is an indicator for any depressive symptoms (CESD) among respondents aged 65 and older. The endogenous regressor is a scalar measure of annual OOP medical spending, and the key policy instrument is an indicator for the post--Medicare Part D period (year $\ge 2006$), during which subsidized prescription drug coverage expanded and catastrophic protection was introduced.

To mirror the simulation designs, we work with two complementary specifications. First, we treat OOP spending in nominal dollars and use a single post--Part D indicator as the excluded instrument. This ``stress test'' design generates a classical weak-IV environment when the risk index is summarized by mean OOP spending, and it provides a direct counterpart to the distributional simulation in Section~\ref{sec:sim}. Second, we express OOP spending in real 2015 dollars using the OECD CPI and show that the same policy indicator becomes a strong instrument in this scaling. The real-dollar specification serves as a well-identified benchmark in which Q–LS and conventional 2SLS can be expected to coincide when we omit a time trend. Once we add a linear trend in survey year, the point estimates remain very similar but Q–LS delivers somewhat tighter standard errors.  Throughout, we trim the top 1\% of the OOP distribution to reduce the influence of extreme outliers.

\subsection{Descriptive evidence on depression and medical spending risk}

\begin{figure}[htbp]
  \centering
  \includegraphics[width=0.9\textwidth]{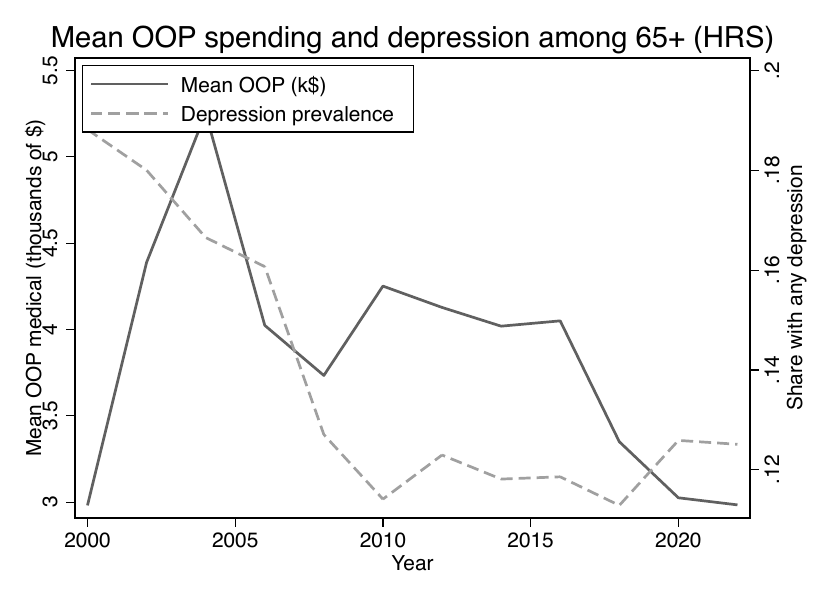}
  \caption{Mean out-of-pocket medical spending and depression prevalence among HRS respondents aged 65+.
  Notes: The figure plots mean annual out-of-pocket (OOP) medical spending (thousands of dollars, left axis)
  and the share of respondents with any depressive symptoms (CESD indicator, right axis) by survey year for
  HRS respondents aged 65 and older. The vertical line at 2006 marks the introduction of Medicare Part D.}
  \label{fig:oop-depression-year}
\end{figure}

\paragraph{Descriptive patterns.}
Figure~\ref{fig:oop-depression-year} plots average out-of-pocket (OOP) medical
spending in thousands of dollars (left axis) and the share of respondents with
any depressive symptoms (right axis) among HRS respondents aged 65 and older
from 2000 to 2014. Mean OOP spending rises from about \$3{,}000 in 2000 to a
peak of roughly \$5{,}200 in the mid-2000s, then falls back toward \$4{,}000–\$4{,}300
in the late 2000s and early 2010s. Over the same period, the prevalence of
depression declines fairly steadily: the share with any depression falls from
around 19–20 percent in 2000 to roughly 11–12 percent by the early 2010s.

These raw trends are consistent with the idea that older households faced
substantial OOP risk in the early 2000s, followed by some compression of
spending risk---coinciding with the implementation of Medicare Part D in
2006---while mental health outcomes improved over time. However, the figure is
purely descriptive: it does not control for changes in the composition of the
65+ population, secular improvements in treatment, or other contemporaneous
policy changes. In the regression analysis below, we therefore focus on
isolating the effect of changes in financial risk exposure, as measured by
OOP spending, using instrumental‐variables methods. In addition, we trim the
regression sample at the 99th percentile of realized annual OOP spending to
remove a small number of extreme outliers; this has negligible impact on the
aggregate time-series patterns in Figure~\ref{fig:oop-depression-year} but
stabilizes the covariate–outcome relationships.

\begin{figure}[htbp]
  \centering
  \includegraphics[width=0.85\textwidth]{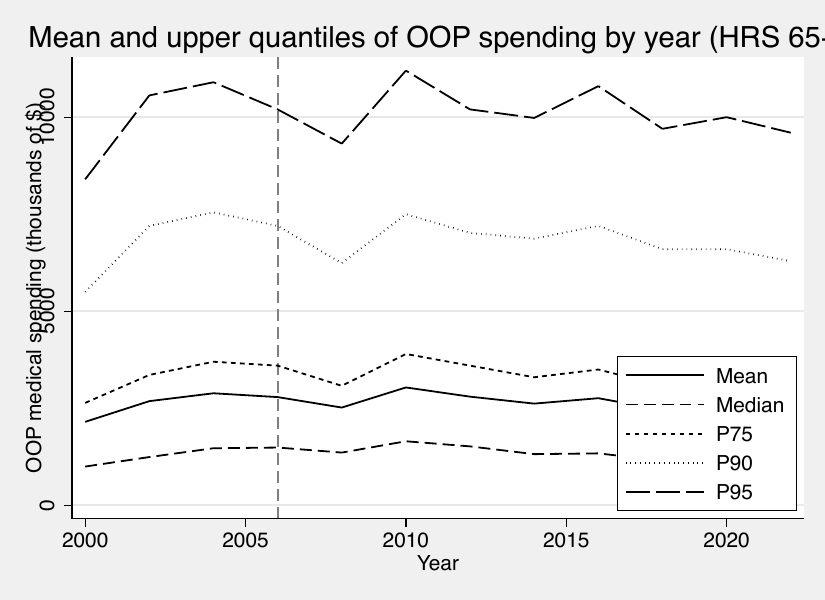}
  \caption{Mean and upper quantiles of out-of-pocket medical spending by year (HRS 65+).
  Notes: The figure plots the mean, median (P50), and upper quantiles (P75, P90, P95)
  of annual OOP medical spending for HRS respondents aged 65 and older,
  from 2000 to 2020. The vertical line at 2006 marks the introduction of Medicare Part D.
  While the median and P75 remain relatively flat over time, the upper quantiles rise in
  the early 2000s and then flatten or decline after 2006, indicating a post–Part D
  compression of the right tail of the OOP spending distribution.}
  \label{fig:oop-quantiles-trend}
\end{figure}

Figure~\ref{fig:oop-quantiles-trend} plots the mean, median, and upper quantiles
(75th, 90th, and 95th percentiles) of annual OOP medical spending for HRS
respondents aged 65 and older between 2000 and 2020. The median and 75th
percentile are remarkably stable over time, indicating that typical OOP
spending changes little over this period. By contrast, the mean closely tracks
movements in the upper tail: the 90th and 95th percentiles rise sharply in the
early 2000s and then flatten or decline following the introduction of
Medicare Part D in 2006. The resulting compression of the gap between the
upper quantiles and the median after 2006 is consistent with Part D primarily
reducing exposure to extreme spending realizations rather than shifting the
center of the distribution. From the perspective of our identification
strategy, this pattern illustrates why standard first-stage diagnostics based
on mean shifts label the post–Part D indicator as a weak instrument for mean
OOP risk, even though it is strongly relevant for the \emph{distribution} of
spending, especially in the right tail.

\subsection{Nominal OOP and a weak-IV design}

We begin with a deliberately ``stressful'' specification that uses nominal OOP
medical spending (in thousands of dollars) as the endogenous regressor and a
single post–Part D indicator as the excluded instrument. This design closely
mirrors the weak-IV, heavy-tailed environments in our Monte Carlo exercises and
is useful for showing how Q--LS behaves when the instrument is weak for the
mean but strongly relevant for the distribution of OOP spending.

Table~\ref{tab:oop-depression} reports linear probability models relating an
indicator for any depressive symptoms to mean OOP medical spending. The sample
consists of 16{,}313 person–year observations for respondents aged 65 and older
in the RAND HRS panel between 2000 and 2020, after excluding person–years with
OOP spending above the 99th percentile of the nominal OOP distribution.

\begin{table}[htbp]\centering
\def\sym#1{\ifmmode^{#1}\else\(^{#1}\)\fi}
\caption{Depression and financial exposure to medical spending (HRS 65+)}
\label{tab:oop-depression}
\begin{tabular}{l*{2}{D{.}{.}{-3}}}
\hline\hline
            &\multicolumn{1}{c}{(1)}&\multicolumn{1}{c}{(2)}\\
            &\multicolumn{1}{c}{OLS}&\multicolumn{1}{c}{2SLS}\\
\hline
OOP medical spending (k\$) $[oop\_med\_k]$
            &       0.002\sym{***}&       0.442         \\
            &     (0.001)         &     (0.410)         \\
[0.75em]
Age         &       0.001\sym{**} &      -0.013         \\
            &     (0.000)         &     (0.013)         \\
[0.5em]
Medicaid    &      -0.032\sym{***}&      -0.224         \\
            &     (0.005)         &     (0.181)         \\
[0.5em]
Medicare    &      -0.156         &      -0.943         \\
            &     (0.098)         &     (0.741)         \\
[0.5em]
Sex (female=1) 
            &      -0.040\sym{***}&       0.045         \\
            &     (0.006)         &     (0.084)         \\
[0.5em]
Hispanic    &       0.094\sym{***}&       0.441         \\
            &     (0.011)         &     (0.330)         \\
\hline
Observations&       16313         &       16313         \\
\hline\hline
\multicolumn{3}{l}{\footnotesize Clustered standard errors in parentheses.}\\
\multicolumn{3}{l}{\footnotesize \sym{*} \(p<0.10\), \sym{**} \(p<0.05\), \sym{***} \(p<0.01\).}\\
\multicolumn{3}{p{0.9\textwidth}}{\footnotesize
\emph{Notes:} The dependent variable is an indicator for any depressive symptoms (CESD).
The main regressor is nominal out-of-pocket medical spending in thousands of dollars
($oop\_med\_k$). The sample consists of 16{,}313 person--year observations for HRS respondents
aged 65 and older between 2000 and 2020, excluding the top 1\% of the nominal OOP spending
distribution. All regressions control for age, Medicaid coverage, Medicare coverage, sex,
and Hispanic ethnicity (shispan), and include a constant. Column (2) reports 2SLS estimates
where $oop\_med\_k$ is instrumented with an indicator for the post--Medicare Part D period
(year $\geq 2006$). Standard errors are clustered at the individual level.}
\end{tabular}
\end{table}

In column~(1), higher nominal OOP spending is positively and statistically
significantly associated with depression: a \$1{,}000 increase in annual OOP
spending is associated with a 0.2 percentage point increase in the probability
of reporting depressive symptoms ($\hat\beta = 0.002$, s.e.\ 0.001). The 2SLS
specification in column~(2), which instruments mean OOP spending with
post–Part D, yields a much larger point estimate ($\hat\beta = 0.442$) but an
even larger standard error (s.e.\ 0.410). The IV estimate is therefore
imprecise and not statistically different from zero at conventional levels.

Table~\ref{tab:firststage} shows that this lack of precision reflects the fact
that the post–Part D indicator is a very weak instrument for mean nominal OOP
spending in the trimmed sample: the cluster-robust partial $F$-statistic for
the excluded instrument is 1.19, far below usual thresholds.

\begin{table}[htbp]\centering
\def\sym#1{\ifmmode^{#1}\else\(^{#1}\)\fi}
\caption{First-stage regression and weak-IV test (mean nominal OOP spending)}
\label{tab:firststage}
\begin{tabular}{l*{1}{D{.}{.}{-3}}}
\hline\hline
            &\multicolumn{1}{c}{(1)}\\
            &\multicolumn{1}{c}{OOP medical spending (k\$)}\\
\hline
Post--Part D indicator
            &      -0.083         \\
            &     (0.076)         \\
[0.75em]
Age         &       0.032\sym{***}\\
            &     (0.006)         \\
[0.5em]
Medicaid    &       0.446\sym{***}\\
            &     (0.069)         \\
[0.5em]
Medicare    &       1.780\sym{***}\\
            &     (0.215)         \\
[0.5em]
Sex (female=1)
            &      -0.193\sym{**} \\
            &     (0.075)         \\
[0.5em]
Hispanic    &      -0.773\sym{***}\\
            &     (0.098)         \\
\hline
Observations&       16313         \\
Weak-IV $F$ (Post--Part D)
            &        1.19         \\
p-value     &       0.275         \\
\hline\hline
\multicolumn{2}{l}{\footnotesize Standard errors in parentheses.}\\
\multicolumn{2}{l}{\footnotesize \sym{*} \(p<0.10\), \sym{**} \(p<0.05\), \sym{***} \(p<0.01\).}\\
\multicolumn{2}{p{0.9\textwidth}}{\footnotesize
\emph{Notes:} The dependent variable is nominal OOP medical spending in thousands of dollars
($oop\_med\_k$). The regression includes age, Medicaid coverage, Medicare coverage, sex, and
Hispanic ethnicity, as well as a constant. The sample consists of 16{,}313 person--year observations
for HRS respondents aged 65 and older between 2000 and 2020, excluding the top 1\% of the nominal
OOP spending distribution. Standard errors are clustered at the individual level. ``Weak-IV $F$
(Post--Part D)'' reports the cluster-robust first-stage $F$-statistic for the excluded instrument
from the corresponding 2SLS specification in Table~\ref{tab:oop-depression}, along with its $p$-value.}
\end{tabular}
\end{table}

\paragraph{Q--LS vs 2SLS in the weak nominal design.}
We now use this weak-IV nominal specification as a stress test to compare Q--LS
and linear 2SLS. Our goal here is diagnostic: to document how the two
estimators behave when the instrument is weak for the mean, yet strongly
relevant for the distribution of OOP spending.

Table \ref{tab:qls_vs_2sls_trim} compares the Q--LS estimate based on the
projected OOP risk index with the conventional linear 2SLS estimate that uses
mean nominal OOP spending as the endogenous regressor.

\begin{table}[htbp]\centering
\def\sym#1{\ifmmode^{#1}\else\(^{#1}\)\fi}
\caption{Q--LS vs linear 2SLS: effect of nominal OOP risk on depression}
\label{tab:qls_vs_2sls_trim}
\begin{tabular}{l*{2}{D{.}{.}{-3}}}
\hline\hline
                    &\multicolumn{1}{c}{(1)}&\multicolumn{1}{c}{(2)}\\
                    &\multicolumn{1}{c}{Q--LS}&\multicolumn{1}{c}{2SLS}\\
\hline
Projected nominal OOP (Q--LS index, k\$)
                    &       0.442\sym{***}&                     \\
                    &     (0.071)         &                     \\
[0.75em]
Nominal OOP spending (k\$)
                    &                     &       0.442         \\
                    &                     &     (0.410)         \\
\hline
Observations        &       16313         &       16313         \\
\hline\hline
\multicolumn{3}{l}{\footnotesize Standard errors in parentheses.}\\
\multicolumn{3}{l}{\footnotesize \sym{*} \(p<0.10\), \sym{**} \(p<0.05\), \sym{***} \(p<0.01\).}\\
\multicolumn{3}{p{0.9\textwidth}}{\footnotesize
\emph{Notes:} The dependent variable is an indicator for any depressive symptoms (CESD).
Column (1) reports the Q--LS estimate obtained by regressing depression on the projected nominal
OOP index $x_{\mathrm{QLS}}$ (in thousands of dollars) and controls, treating $x_{\mathrm{QLS}}$
as a generated regressor and using cluster-robust standard errors. Column (2) reports the conventional
2SLS estimate using nominal mean OOP spending as the endogenous regressor instrumented with the
post--Part D indicator. All specifications control for age, Medicaid coverage, Medicare coverage,
sex, and Hispanic ethnicity, and include a constant. The sample consists of 16{,}313 person--year
observations between 2000 and 2020, trimming the top 1\% of the nominal OOP spending distribution.
Standard errors are clustered at the individual level.}
\end{tabular}
\end{table}

The point estimates in columns~(1) and (2) are numerically identical (0.442),
but the standard errors differ dramatically. In the weak-IV nominal
specification, Q--LS delivers a standard error of 0.071 and a precisely
estimated positive effect, whereas linear 2SLS yields a standard error of
0.410 and a statistically insignificant estimate. This pattern mirrors our
simulation results: when the policy instrument reshapes the distribution of
OOP spending but produces only a noisy mean shift, Q--LS can extract a much
stronger effective first stage from the same variation, tightening
confidence intervals without changing the underlying signal.

In this nominal specification, the combination of weak mean
relevance and heavy tails also makes the estimated magnitude difficult to
interpret substantively. In the next subsection, we therefore turn to a more
natural specification that expresses OOP spending in real terms.

\subsection{Real OOP: from strong-IV benchmark to trend-adjusted design}

For substantive interpretation, it is more natural to express OOP spending in
real dollars. In the second part of the empirical section we therefore: (i)
redefine OOP spending in thousands of 2015 U.S.\ dollars, and (ii) show how
the Part D instrument looks first like a strong mean instrument, and then,
after soaking up secular trends, like a primarily \emph{distributional}
instrument where Q--LS mainly improves precision.

We deflate nominal OOP using the OECD consumer price index obtained from FRED,
constructing a monthly index with base year 2015 and aggregating to the
relevant survey years. Real OOP spending in thousands of 2015 dollars is
denoted $oop\_med\_2015\_k$.\footnote{Organization for Economic Co-operation
and Development, \emph{Consumer Price Indices (CPIs, HICPs), COICOP 1999:
Consumer Price Index: Total for United States}, retrieved from FRED, Federal
Reserve Bank of St. Louis; \url{https://fred.stlouisfed.org/series/USACPIALLMINMEI},
accessed January 14, 2026.}

\paragraph{Real OOP without time trend: a strong-IV benchmark.}
In the first real-dollar specification, we repeat the nominal regressions
replacing $oop\_med\_k$ with $oop\_med\_2015\_k$ and keeping the same controls.
In this case, the post–Part D indicator is a strong instrument for real OOP:
the first-stage partial $F$-statistic is around 50 in the trimmed sample. Table
\ref{tab:qls_vs_2sls_real} compares Q--LS and linear 2SLS for the depression
outcome in this strong-IV environment.

\begin{table}[htbp]\centering
\def\sym#1{\ifmmode^{#1}\else\(^{#1}\)\fi}
\caption{Q--LS vs linear 2SLS: effect of real OOP risk on depression (2015\$)}
\label{tab:qls_vs_2sls_real}
\begin{tabular}{l*{2}{D{.}{.}{-3}}}
\hline\hline
                    &\multicolumn{1}{c}{(1)}&\multicolumn{1}{c}{(2)}\\
                    &\multicolumn{1}{c}{Q--LS}&\multicolumn{1}{c}{2SLS}\\
\hline
Projected real OOP (Q--LS index, 2015 k\$)
                    &       0.045\sym{***}&                     \\
                    &     (0.007)         &                     \\
[0.75em]
Real OOP spending (2015 k\$)
                    &                     &       0.045\sym{***}\\
                    &                     &     (0.008)         \\
\hline
Observations        &       16311         &       16311         \\
\hline\hline
\multicolumn{3}{l}{\footnotesize Standard errors in parentheses.}\\
\multicolumn{3}{l}{\footnotesize \sym{*} \(p<0.10\), \sym{**} \(p<0.05\), \sym{***} \(p<0.01\).}\\
\multicolumn{3}{p{0.85\textwidth}}{\footnotesize
\emph{Notes:} The dependent variable is an indicator for any depressive symptoms (CESD).
In column (1), the endogenous regressor is the projected real OOP index in thousands
of 2015 dollars constructed via the Q--LS procedure; the coefficient is estimated by
OLS with cluster-robust standard errors. In column (2), real OOP spending in thousands
of 2015 dollars is instrumented with the post--Part D indicator, and the coefficient
is estimated by linear 2SLS with cluster-robust standard errors. All regressions control
for age, Medicaid, Medicare, gender, and Hispanic origin. OOP spending is deflated to
2015 dollars using the OECD CPI (index 2015=100), and the sample trims the top 1\% of
the real OOP distribution. Standard errors are clustered at the individual level.}
\end{tabular}
\end{table}

In this specification, Q--LS and 2SLS are virtually indistinguishable: both
yield an effect of about 0.045 per \$1{,}000 of real OOP, with very similar
standard errors. This is exactly the ``well behaved'' case in which our theory
predicts that Q--LS collapses to optimal 2SLS when the instrument is strong and
the linear mean representation is adequate. Table~\ref{tab:qls_se_compare}
, in Appendix, confirms that three different ways of evaluating uncertainty for the Q--LS
coefficient---cluster-robust OLS, clustered bootstrap, and the analytical
Q--LS sandwich formula---deliver nearly identical standard errors.
\paragraph{Real OOP with time trend: a distributional relevance design.}
Finally, we add a flexible linear time trend to the
real-OOP specification. This absorbs much of the secular drift in spending and
returns us to a situation in which post--Part D induces rich distributional
changes in real OOP but only a small additional mean shift conditional on time trend: with first-stage partial $F$-statistic is around 2. 
From the perspective of classical diagnostics, the instrument again looks
fragile; from the perspective of distributional relevance, it remains 
informative.

Table~\ref{tab:qls_vs_2sls_real_trend} reports Q--LS and 2SLS estimates for
this trend-adjusted real specification.

\begin{table}[htbp]\centering
\def\sym#1{\ifmmode^{#1}\else\(^{#1}\)\fi}
\caption{Q--LS vs linear 2SLS: effect of real OOP risk on depression with time trend (2015\$)}
\label{tab:qls_vs_2sls_real_trend}
\begin{tabular}{l*{2}{D{.}{.}{-3}}}
\hline\hline
                    & \multicolumn{1}{c}{(1)} & \multicolumn{1}{c}{(2)} \\
                    & \multicolumn{1}{c}{Q--LS} & \multicolumn{1}{c}{2SLS} \\
\hline
Projected real OOP (Q--LS index, 2015 k\$)
                    &   0.07                  &                         \\
                    & (0.053)                  &                         \\
[0.75em]
Real OOP spending (2015 k\$)
                    &                          &   0.07                \\
                    &                          & (0.079)                 \\
\hline
Observations        & \multicolumn{1}{c}{16{,}295} & \multicolumn{1}{c}{16{,}295} \\
\hline\hline
\multicolumn{3}{l}{\footnotesize Standard errors in parentheses.}\\
\multicolumn{3}{l}{\footnotesize \sym{*} \(p<0.10\), \sym{**} \(p<0.05\), \sym{***} \(p<0.01\).}\\
\multicolumn{3}{p{0.9\textwidth}}{\footnotesize
\emph{Notes:} The dependent variable is an indicator for any depressive symptoms (CESD). 
Column (1) reports the Q--LS estimate obtained by regressing depression on the projected real
OOP index $x_{\mathrm{QLS}}$ (in thousands of 2015 dollars), treating $x_{\mathrm{QLS}}$ as a
generated regressor and using cluster-robust standard errors. Column (2) reports the linear 2SLS
estimate where real OOP medical spending (in thousands of 2015 dollars) is instrumented with the
post--Part D indicator. All specifications control for age, Medicaid coverage, Medicare coverage,
sex, Hispanic origin, a linear time trend in years since 2000, and race (sracem), and include a
constant. The sample consists of 16{,}295 person--year observations for HRS respondents aged 65
and older between 2000 and 2020, trimming the top 1\% of the real OOP distribution. Standard errors
are clustered at the individual level.}
\end{tabular}
\end{table}

In this trend-adjusted real specification, Q--LS and 2SLS again deliver the
same point estimate (0.07 per \$1{,}000 of real OOP). The main difference is
precision: the Q--LS standard error is 0.05, whereas the 2SLS standard error
is 0.07. Thus Q--LS reduces the confidence interval width by roughly one-third,
but the estimate remains only marginally significant at best (p-value around
0.14). Economically, the estimates across the  real-dollar specifications
suggest a modest positive effect of OOP risk on depression, with the strongest
evidence in the specification that treats Part D as a strong mean instrument
and somewhat weaker evidence once secular trends are absorbed.

Taken together, the three designs -- nominal weak-IV, real strong-IV, and
real--with-trend -illustrate the two key empirical messages of the paper.
First, the way in which a policy reform is summarized in the first stage (nominal vs real,
with or without trend) can dramatically change whether the design looks weak
or strong in mean terms, even when the underlying distributional shifts are
similar. Second, when mean relevance is fragile but distributional relevance
can be strong, Q--LS behaves like a natural complement to 2SLS: it recovers the
same underlying signal where the design is well identified, and it delivers
substantial precision gains relative to linear IV in settings where standard
weak-IV diagnostics based on mean shifts understate the information content of
the instrument.

\section{Concluding remarks}
\label{sec:conc}
Our results can be summarized along two dimensions. On the identification side, we formalize a notion of \emph{distributional relevance} and show that even \emph{purely distributional} instruments---for which $\Var(\E[X\mid Z]) = 0$ while $F_{X\mid Z}(\cdot\mid Z)$ varies nontrivially in $Z$---are sufficient to identify average structural functions in a nonseparable triangular model via the control function $V = F_{X\mid Z}(X\mid Z)$. This makes explicit how instruments that look weak in the mean can still be strong in the distribution, and connects the control–function identification argument to weak-IV concerns in linear models.

On the estimation side, we propose a simple, implementable estimator, Quantile Least Squares (Q--LS), which constructs an optimal distributional instrument from conditional quantiles and then uses it in a conventional linear IV step. In designs where policy reforms are primarily meant to reshape the distribution of risk rather than to move its mean, Q--LS complements the existing 2SLS and weak-IV toolkits by extracting a stronger effective first stage from the same policy variation, while leaving classical strong-IV applications essentially unchanged.

Q--LS exploits this structure by choosing, within a given dictionary of
quantile-based functions, the instrument that best predicts $X$ in mean
square error and then using it in a standard 2SLS procedure. We showed that
Q--LS is consistent and asymptotically normal under distributional relevance,
that it coincides with optimal 2SLS in Gaussian/location settings where mean
relevance holds, and that it remains reliable when classical first-stage
diagnostics indicate weak instruments from the perspective of the conditional
mean. We also developed ridge and LASSO regularisation schemes that stabilise
the first stage and mitigate many-instrument bias, while preserving the
first-order asymptotics of Q--LS and providing practically useful variants
(Q–LS--R, Q–LS--L1, Q–LS--L2) with distinct robustness properties.

Our Monte Carlo evidence highlights three main messages. First, when the
instrument primarily shifts the tails of the $X$-distribution, Q--LS can
recover the same limiting parameter as optimal 2SLS but with substantially
smaller finite-sample variance, even in designs where the mean first stage is
formally weak. Second, the regularised Q--LS variants are effective at
controlling the dispersion of the estimator in many-instrument designs and
under moderate misspecification of the quantile representation. Third,
conventional IV estimators can behave erratically in these designs, with
large RMSE and unstable confidence intervals, even when coverage remains
nominal, whereas Q--LS delivers more concentrated sampling distributions.

The empirical application to Medicare Part D illustrates these features in a
setting where policy reforms plausibly operate through the \emph{distribution}
of medical spending risk rather than its mean. In the nominal-dollar
specifications, standard first-stage diagnostics classify the post--Part D
indicator as a very weak instrument for mean OOP spending, and conventional
2SLS delivers large but imprecise estimates of the effect of OOP risk on
depression. By contrast, Q--LS exploits the pronounced
compression of the upper tail of the OOP distribution to construct a strong
distributional first stage from the same policy variation, yielding
substantially tighter confidence intervals while preserving the economic
interpretation of the effect. When we deflate spending to 2015 dollars, the
mean first stage becomes much stronger and conventional 2SLS performs well;
in that case Q--LS and 2SLS again agree closely on the point estimates, but
Q--LS offers at most modest efficiency gains. Taken together, the simulations
and the Medicare Part D evidence suggest viewing Q--LS as a complement to
standard IV: it nests the optimal mean-based instrument in classical designs,
but remains informative in applications where instruments are weak in the mean
and strong in the distribution.

The distributional relevance framework extends beyond linear IV. In an appendix, we show that purely distributional instruments can identify average structural functions in nonseparable triangular models via a control-function
representation based on the conditional distribution $F_X(X\mid Z)$. This
connects our approach to the literature on nonparametric IV and completeness,
and suggests several directions for future work: deriving more primitive
conditions that guarantee nondegenerate optimal instruments; developing
formal diagnostics for the linear quantile representation and for
instrument-instability; designing data-driven procedures for selecting
quantile grids and penalties; and applying Q--LS in empirical settings where
policy or institutional variation plausibly operates through changes in risk
exposure rather than simple mean shifts.

\newpage

\bibliographystyle{agsm} % or try abbrvnat or unsrtnat
\bibliography{reference.bib} % refers to example.bib
\appendix

% Reset counters
\setcounter{assumption}{0}
\setcounter{lemma}{0}
\setcounter{prop}{0}
\setcounter{remark}{0}
\setcounter{theorem}{0}
% Appendix-style numbering
\renewcommand{\theassumption}{A.\arabic{assumption}}
\renewcommand{\thelemma}{A.\arabic{lemma}}
\renewcommand{\theprop}{A.\arabic{prop}}

\section{Nonseparable triangular models and distributional relevance}\label{sec:app_nstm}

%\addcontentsline{toc}{section}{Appendix A: Nonseparable triangular models and distributional relevance}

This appendix collects the extension of distributional relevance and the
control-function approach to nonseparable triangular models, and outlines a
semiparametric estimator for the average structural function. The material
follows the draft text in Sections~\ref{sec:model}--\ref{sec:semiparametric}
of the earlier version of the paper, with notational adjustments for
consistency.

%\section{Model, control function, and instrument relevance beyond the mean}
\label{sec:model}

We build on the nonseparable triangular framework of
\citet{cfnsv20_cfa} to formalize a notion of \emph{distributional
relevance} of instruments. Our goal is to show that identification of the
average structural function can be achieved even when instruments are
``weak'' in the classical mean sense, provided they are sufficiently strong in
a distributional sense.

\subsection{Triangular model and object of interest}

Let $(Y,X,Z)$ be random variables defined on a common probability space.
We consider the following nonseparable triangular system:
\begin{align}
    Y &= f(X,Z_1,\varepsilon), \label{eq:structural-np}\\
    X &= g(Z,\eta), \label{eq:firststage-np}
\end{align}
where $f$ and $g$ are unknown measurable functions, and
$\varepsilon,\eta$ are unobserved disturbances. The variable $Z=(Z_1,Z_2)'$
is a vector of instruments.

\begin{assumption}[Exogeneity and normalization]
\label{ass:exog}
The pair $(\varepsilon,\eta)$ is independent of $Z$, and
$\eta \sim U(0,1)$.
\end{assumption}

We are interested in the \emph{average structural function} (ASF)
\begin{equation}
    \mu(x,z_1) := \E\big[ f(x,z_1,\varepsilon) \big],
    \label{eq:ASF-def}
\end{equation}
which captures the mean potential outcome if the endogenous and exogenous
regressors were set to the value $x$ and $z_1$.

\subsection{Control-function representation via the conditional distribution of $X$}

Let $F_X(x\mid z)$ denote the conditional distribution function of $X$ given
$Z=z$:
\[
    F_X(x \mid z) := \p(X \le x \mid Z=z).
\]
Define the \emph{control variable}
\begin{equation}
    V := F_X(X \mid Z).
    \label{eq:controlV-def}
\end{equation}
Following \citet{cfnsv20_cfa}, we impose a monotonicity condition on the
first stage.

\begin{assumption}[Monotone first stage]
\label{ass:monotone}
For each $z$ in the support of $Z$, the map
$\eta \mapsto g(z,\eta)$ in \eqref{eq:firststage-np} is strictly increasing
and continuous.
\end{assumption}

Under Assumptions~\ref{ass:exog}--\ref{ass:monotone}, the first stage
admits a conditional quantile representation; see
\citet{cfnsv20_cfa}:
\[
    X = g(Z,\eta)
    \quad \Longrightarrow \quad
    F_X(X \mid Z) = \eta \quad \text{a.s.,}
\]
so that
\begin{equation}
    V = F_X(X \mid Z) = \eta \quad \text{a.s.}
    \label{eq:V-equals-eta}
\end{equation}
In particular, $V$ is independent of $Z$ and can be used as a control
function for the endogeneity of $X$.

A key implication of \eqref{eq:V-equals-eta} and
Assumption~\ref{ass:exog} is that $\varepsilon$ is independent of $(X,Z)$
conditional on $V$:
\begin{equation}
    \varepsilon \perp (X,Z) \mid V.
    \label{eq:epsilon-indep}
\end{equation}
It follows from \eqref{eq:structural-np} and \eqref{eq:epsilon-indep} that
\begin{equation}
    \E[ Y \mid X=x, Z_1=z_1, V=v ]
    = \E[ f(x,z_1,\varepsilon) \mid V=v ].
    \label{eq:cond-mean-Y-given-XV}
\end{equation}

To recover the ASF $\mu(x,z_1)$ in \eqref{eq:ASF-def}, we rely on a support
condition that is standard in control-function identification results.

\begin{assumption}[Support / overlap]
\label{ass:support}
For each $x$ in the interior of the support of $X$, the conditional support
of $V$ given $X=x$ coincides with the marginal support of $V$:
\[
    \supp(V \mid X=x) = \supp(V).
\]
\end{assumption}

Under Assumptions~\ref{ass:exog}--\ref{ass:support}, one obtains the usual
control-function representation for the ASF:
\begin{equation}
    \mu(x,z_1)
    = \E[f(x,z_1,\varepsilon)]
    = \int \E[Y \mid X=x,Z_1=z_1, V=v] \, dF_V(v).
    \label{eq:ASF-CF-repr}
\end{equation}

\subsection{Mean relevance and distributional relevance in an $L^2$ space}

The notions of mean relevance and distributional relevance in
Definitions~\ref{def:mean-relevance} and \ref{def:dist-relevance} extend
verbatim to this nonparametric setting: we view $m(z)=\E[X\mid Z=z]$ as an
element of $L^2(P_Z)$, and $F_X(\cdot\mid z)$ as an element of
$L^2(\lambda)$. We then define purely distributional instruments as in
Definition~\ref{def:purely-dist}.

\subsection{Identification of the ASF under purely distributional relevance}

We now show that the ASF in \eqref{eq:ASF-def} remains nonparametrically
identified under purely distributional relevance, even though the
instruments are completely irrelevant for the conditional mean of $X$.

\begin{theorem}[ASF identification with purely distributional instruments]
\label{thm:ASF-ident-dist}
Suppose Assumptions~\ref{ass:exog}, \ref{ass:monotone}, and
\ref{ass:support} hold. Assume further that $Z$ includes only purely
distributional instruments for $X$ in the sense of
Definition~\ref{def:purely-dist}. Then, for every $x$ in the interior of the
support of $X$, the ASF $\mu(x,z_1) = \E[f(x,z_1,\varepsilon)]$ is
nonparametrically identified and admits the representation
\begin{equation}
    \mu(x,z_1)
    = \int \E[ Y \mid X=x,Z_1=z_1, V=v ] \, dF_V(v).
    \label{eq:ASF-ident-dist}
\end{equation}
In particular, identification of $\mu(x,z_1)$ does not require any mean
relevance of the instruments; it is enough that $Z$ is distributionally
relevant for $X$.
\end{theorem}

\begin{proof}
By Assumptions~\ref{ass:exog} and \ref{ass:monotone}, we have
$V = F_X(X\mid Z) = \eta$ almost surely, and $V$ is independent of $Z$.
Moreover, the joint system
\eqref{eq:structural-np}--\eqref{eq:firststage-np} with
$(\varepsilon,\eta)\perp Z$ implies that $\varepsilon$ is independent of
$(X,Z)$ conditional on $V$, as in \eqref{eq:epsilon-indep}. Therefore,
\[
    \E[ Y \mid X=x,Z_1=z_1, V=v ]
    = \E[ f(x,z_1,\varepsilon) \mid V=v ],
\]
as in \eqref{eq:cond-mean-Y-given-XV}. By Assumption~\ref{ass:support}, the
conditional support of $V$ given $X=x$ coincides with its marginal support.
Using the law of total expectation, we can write
\[
    \mu(x,z_1)
    = \E[ f(x,z_1,\varepsilon) ]
    = \int \E[ f(x,z_1,\varepsilon) \mid V=v ] \, dF_V(v)
    = \int \E[ Y \mid X=x,Z_1=z_1, V=v ] \, dF_V(v),
\]
which yields \eqref{eq:ASF-ident-dist}. The purely distributional
instrument condition ensures that $F_X(\cdot\mid Z)$ varies with $Z$ in an
$L^2(\lambda)$ sense, so that the control variable $V = F_X(X\mid Z)$ is
non-degenerate. The lack of mean relevance, $\Var(\E[X\mid Z])=0$, does not
affect this argument, since the control-function approach does not require
$Z$ to move the mean of $X$.
\end{proof}

\subsection{Semiparametric estimation with distributionally relevant instruments}
\label{sec:semiparametric}

We briefly outline a semiparametric estimator of the ASF based on a
two-step plug-in approach: (i) estimate the control function
$V = F_X(X\mid Z)$, and (ii) estimate the conditional mean function
$m(x,z_1,v) := \E[Y\mid X=x,Z_1=z_1,V=v]$ using flexible basis expansions.
The estimator of the ASF is then obtained by plugging in the estimated
control function and conditional mean.

Let $\{(Y_i,X_i,Z_i)\}_{i=1}^n$ be an i.i.d.\ sample. The estimator proceeds
as follows.

\begin{enumerate}
    \item[\textbf{Step 1.}] Estimate the conditional distribution (or
    quantile process) of $X$ given $Z$ and construct an estimate
    $\hat V_i$ of the control variable $V_i = F_X(X_i \mid Z_i)$.
    \item[\textbf{Step 2.}] Regress $Y$ on a flexible series expansion in
    $(X,Z_1,\hat V)$ to estimate the conditional mean function
    $m(x,z_1,v)$, and use the control-function representation
    \[
        \mu(x,z_1) = \int m(x,z_1,v)\, dF_V(v)
    \]
    to construct a plug-in estimator $\hat \mu(x,z_1)$.
\end{enumerate}

The first stage can be implemented using series quantile regression as in
\citet{cfk15_qiv} and \citet{cfnsv20_cfa}. A finite grid of quantile
indices $\mathcal{T}_n=\{\tau_1,\dots,\tau_{K_n}\}$ is used to estimate
conditional quantiles $Q_X(\tau_k\mid Z)$ by series quantile regression in
$Z$, and the conditional CDF at $(X_i,Z_i)$ is approximated by
\begin{equation}
    \hat F_X(X_i \mid Z_i)
    :=  \mathbbm{1}\big\{ \hat Q_X(\tau_k \mid Z_i) \le X_i \big\}.
    \label{eq:FXhat-quantile-app}
\end{equation}
The resulting estimated control variable is $\hat V_i = \hat F_X(X_i\mid Z_i)$.

In the second stage, one specifies a basis $w(X,Z_1,V)$ in $(X,Z_1,V)$ and
estimates
\[
    m(x,z_1,v) \approx w(x,z_1,v)'\beta
\]
by least squares, replacing $V_i$ by $\hat V_i$. The ASF is then estimated
by
\begin{equation}
    \hat\mu(x,z_1)
    := \frac{1}{n} \sum_{i=1}^n \hat m(x,z_1,\hat V_i)
    = \frac{1}{n} \sum_{i=1}^n w(x,z_1,\hat V_i)' \hat\beta.
    \label{eq:mu-hat-app}
\end{equation}

Under standard smoothness and sieve-approximation assumptions on
$F_X(x\mid z)$ and $m(x,z_1,v)$, and suitable growth rates for the sieve
dimensions and the number of quantiles $K_n$, one can show that
$\hat\mu(x,z_1)$ is consistent and asymptotically normal, with an influence
function that is orthogonal to the first-stage estimation error in $\hat V$;
see \citet{cfnsv20_cfa} for related results.

% =======================================================
\section{A LATE interpretation for Q--LS under a binary treatment}\label{app:late}

This appendix provides a potential-outcome interpretation of Q--LS as a LATE estimator
under additional structure. The main text treats Q--LS primarily as a strength/efficiency
device for IV. Here we show that when the endogenous variable is binary (or can be
interpreted as a binary treatment derived from a latent index), Q--LS naturally delivers a
\emph{scalar monotone instrument} for which the standard LATE logic applies.

\subsection{Setup: binary treatment, scalar instrument, and a Q--LS index}\label{app:late:setup}

Let $D\in\{0,1\}$ be the endogenous treatment and let $Y(d)$ denote potential outcomes
for $d\in\{0,1\}$. We observe $Y = Y(D)$. For the LATE interpretation in this appendix,
we specialize to a \emph{scalar} instrument $Z$ (which can be binary, $Z\in\{0,1\}$,
or more generally ordered). If exogenous covariates $Z_1$ are present, we work after
partialling them out and suppress $Z_1$ in the notation.

In the quantile-dictionary implementation of Q--LS, we consider a finite grid
$\mathcal{T} = \{\tau_1,\ldots,\tau_K\} \subset (0,1)$ and define
\[
  G(Z)
  :=
  \big(
    Q_{D\mid Z}(\tau_1\mid Z),
    \ldots,
    Q_{D\mid Z}(\tau_K\mid Z)
  \big)',
\]
where $Q_{D\mid Z}(\tau\mid Z)$ denotes the conditional $\tau$-quantile of $D$ given $Z$.
Q--LS then forms scalar indices of the form
\[
  H_w(Z) := G(Z)'w,
  \qquad w \in \Delta_K := \big\{w\in\mathbb{R}^K:\ w_k\ge 0,\ \mathbf{1}'w=1\big\}.
\]
When $w\in\Delta_K$, $H_w(Z)$ is a convex combination of conditional quantiles---an
L-statistic---which is useful both for monotonicity and for interpretability.

\begin{assumption}[IV validity]\label{ass:late:valid}
(\emph{Independence and exclusion})
The instrument $Z$ is independent of the potential outcomes and potential treatments,
and affects $Y$ only through $D$:
\[
  \{Y(0),Y(1),D(z)\}_{z} \;\perp\!\!\!\perp\; Z.
\]
(If covariates $Z_1$ are included, the independence is conditional on $Z_1$.)
\end{assumption}

\begin{assumption}[Quantile relevance]\label{ass:late:qrel}
(\emph{Distributional relevance of the scalar instrument})
There exists at least one $\tau_k\in\mathcal{T}$ such that the conditional quantile
$Q_{D\mid Z}(\tau_k\mid Z)$ is not almost surely constant in $Z$. Equivalently, not all
components of $G(Z)$ are almost surely flat in $Z$.
\end{assumption}

\begin{remark}
When $D$ is binary, the conditional quantiles $Q_{D\mid Z}(\tau\mid Z)$ are step functions
of the propensity score $p(Z):=\Pr(D=1\mid Z)$. In this case, Assumption~\ref{ass:late:qrel}
is equivalent to $p(Z)$ not being almost surely constant. We formulate relevance in terms of
conditional quantiles to keep the connection with the distributional-instrument framework
used in the main text.
\end{remark}

Quantile relevance ensures that there exist weights $w\in\Delta_K$ for which
$\mathrm{Cov}\big(H_w(Z),D\big)\neq 0$, i.e.\ the scalar index $H_w$ is a relevant instrument
for $D$.

\begin{assumption}[Quantile monotonicity in the scalar instrument]\label{ass:late:qmonoZ}
For each $\tau\in\mathcal{T}$, the conditional quantile $z\mapsto Q_{D\mid Z}(\tau\mid z)$
is weakly increasing in $z$.
\end{assumption}

\begin{lemma}[Convex combinations of quantiles preserve monotonicity]\label{lem:qmono_scalar}
Under Assumption~\ref{ass:late:qmonoZ}, any index $H_w(Z)=G(Z)'w$ with $w\in\Delta_K$
is weakly increasing in the scalar instrument $Z$.
\end{lemma}

\begin{proof}
For each $\tau_k\in\mathcal{T}$, $z\mapsto Q_{D\mid Z}(\tau_k\mid z)$ is weakly increasing by
Assumption~\ref{ass:late:qmonoZ}. A convex combination of weakly increasing functions is
weakly increasing, so $z\mapsto H_w(z)=\sum_{k=1}^K w_k Q_{D\mid Z}(\tau_k\mid z)$ is also
weakly increasing.
\end{proof}

In what follows, we fix a weight vector $w^\star\in\Delta_K$ (e.g.\ the one selected by Q--LS)
that satisfies Assumption~\ref{ass:late:qrel}, and write $H:=H_{w^\star}(Z)$.

\begin{assumption}[Monotone treatment selection w.r.t.\ the index]\label{ass:late:mono}
(\emph{No defiers in the index ordering})
There exists a version of the potential treatment process $\{D(h)\}_{h}$ such that
\[
  h_1 \ge h_0 \quad\Longrightarrow\quad D(h_1) \ge D(h_0)\qquad\text{a.s.}
\]
In particular, if $H_1 \ge H_0$ then $D(H_1)\ge D(H_0)$ almost surely.
\end{assumption}

Assumption~\ref{ass:late:mono} is the usual no-defiers (monotonicity) condition, expressed with respect to the scalar Q--LS index $H$. Assumption \ref{ass:late:qmonoZ} and Lemma \ref{lem:qmono_scalar} justify interpreting $H$ as a one-dimensional ``encouragement" variable inherited from the scalar instrument $Z$.

\subsection{Wald LATE using the Q--LS index}\label{app:late:wald}

When $H$ takes at least two values $h_0<h_1$ with positive probability, the standard Wald ratio
based on those two values identifies a LATE for compliers defined relative to $H$.

\begin{prop}[LATE for two index values]\label{prop:late:two}
Suppose Assumptions~\ref{ass:late:valid}, \ref{ass:late:qrel},
\ref{ass:late:qmonoZ}, and \ref{ass:late:mono} hold. Let $H$ take the values
$h_0<h_1$ in its support with $\Pr(H=h_j)>0$ for $j\in\{0,1\}$ and
$\E[D\mid H=h_1]\neq \E[D\mid H=h_0]$. Then
\begin{equation}\label{eq:late:two}
\frac{\E[Y\mid H=h_1]-\E[Y\mid H=h_0]}{\E[D\mid H=h_1]-\E[D\mid H=h_0]}
=
\E\!\left[Y(1)-Y(0)\ \big|\ D(h_1)>D(h_0)\right].
\end{equation}
\end{prop}

\noindent
\emph{Proof.}
Under IV validity, $\E[Y\mid H=h]=\E[Y(D(h))]$ and $\E[D\mid H=h]=\E[D(h)]$.
Write $Y(D(h)) = Y(0) + D(h)\{Y(1)-Y(0)\}$ and take differences between $h_1$ and
$h_0$. Monotonicity implies $D(h_1)-D(h_0)\in\{0,1\}$, which delivers
\eqref{eq:late:two}. \qed

\subsection{2SLS with a multi-valued scalar index: positive LATE weights}\label{app:late:2sls}

In practice, Q--LS produces a scalar index $H$ with many support points, and empirical work
typically uses 2SLS with $H$ (possibly along with controls). While multiple instruments can
generate non-convex (possibly negative) weights across different LATEs, a key advantage of using a
\emph{single scalar} index is that the IV estimand can be expressed as a \emph{positive-weighted}
average of LATEs across index increments.

Assume $H$ has finite support $\{h_1<\cdots<h_J\}$ for exposition, and define adjacent-index LATEs
\[
  \text{LATE}_j
  := \E\!\left[Y(1)-Y(0)\ \big|\ D(h_{j+1})>D(h_j)\right],
  \qquad j=1,\ldots,J-1.
\]
Let $\Delta_j := \E[D\mid H=h_{j+1}] - \E[D\mid H=h_j] \ge 0$ under
Assumption~\ref{ass:late:mono}.

\begin{prop}[2SLS as a positive average of adjacent LATEs]\label{prop:late:avg}
Suppose Assumptions~\ref{ass:late:valid}, \ref{ass:late:qrel},
\ref{ass:late:qmonoZ}, and \ref{ass:late:mono} hold and $H$ takes values
$\{h_1<\cdots<h_J\}$ with positive probability. Then the population IV
coefficient using the single instrument $H$,
\[
  \beta_{IV}(H):=\frac{\mathrm{Cov}(H,Y)}{\mathrm{Cov}(H,D)},
\]
can be written as
\[
  \beta_{IV}(H) = \sum_{j=1}^{J-1} \delta_j\, \text{LATE}_j,
  \qquad
  \delta_j := \frac{\pi_j}{\sum_{\ell=1}^{J-1}\pi_\ell},
  \qquad
  \pi_j \ge 0,
\]
where $\text{LATE}_j$ is the adjacent-index LATE between $h_j$ and $h_{j+1}$.
In particular, the weights $\delta_j$ are nonnegative and sum to one.
\end{prop}

\noindent
\emph{Remark.}
The exact form of $\pi_j$ depends on the distribution of $H$ (and on whether $H$ is
centered), but the key point is conceptual: \emph{with a single scalar index} satisfying
monotonicity, the IV estimand aggregates local LATEs using \emph{positive weights} tied to
how much the index shifts treatment take-up across adjacent levels.

\subsection{Why Q--LS is helpful for a LATE interpretation}\label{app:late:why}

Propositions~\ref{prop:late:two}--\ref{prop:late:avg} show that the LATE interpretation hinges on
(i) validity of the underlying instrument and (ii) monotonicity with respect to a scalar index.
Q--LS contributes by \emph{constructing} such a scalar index from distributional features of the
first stage.

\paragraph{(i) Scalar reduction.}
Even when the original instrument vector is high-dimensional, the quantile dictionary
$G(Z)$ and the Q--LS weights $w^\star$ collapse information into a one-dimensional index
$H=H_{w^\star}(Z)$. This facilitates monotonicity assumptions and avoids the negative-weight
issues that can arise when combining multiple instruments in 2SLS.

\paragraph{(ii) Positive normalized weights and monotone distributional shifts.}
When $w\in\Delta_K$, $H_w(Z)$ is a convex combination of conditional quantiles. Under
Assumption~\ref{ass:late:qmonoZ}, higher values of the scalar instrument $Z$ shift each
conditional quantile of $D$ weakly to the right, so $H_w(Z)$ inherits this ordering
(Lemma~\ref{lem:qmono_scalar}). This provides a transparent route for arguing that $H$
acts as a monotone ``encouragement'' variable, supporting Assumption~\ref{ass:late:mono}.

\paragraph{(iii) Strength.}
Within the class of admissible indices (e.g.\ $w\in\Delta_K$), Q--LS chooses weights to
maximize predictive content for the first stage. Thus, Q--LS does not change the causal
content of IV validity, but can materially improve finite-sample performance by increasing
$\mathrm{Cov}(H,D)$.

\begin{remark}[Inspecting which complier groups matter most]
The representation in Proposition~\ref{prop:late:avg} is not only
conceptual: in empirical work the weights $\omega_j$ can be estimated
from the data, because they are proportional to the contribution of each
adjacent index increment to $Cov(H,D)$. Concretely, letting
$\hat{\Delta}_j := \E_n[D\mid H=h_{j+1}] - \E_n[D\mid H=h_j]$ and
$\hat{p}_j := \Pr_n(H\in\{h_j,h_{j+1}\})$, we can form empirical
analogues $\hat{\pi}_j$ and $\hat{\delta}_j$ and report them as a
diagnostic. This makes it possible to see which ``complier bands''
between $h_j$ and $h_{j+1}$ contribute most to the overall Q--LS
estimate.
\end{remark}

\begin{remark}[Which parts of the distribution drive the Q--LS index?]
Because the Q--LS index takes the form
\(
H_w(Z) = G(Z)'w
\)
with \(
G(Z) = \big(Q_{X\mid Z}(\tau_1\mid Z),\ldots,Q_{X\mid Z}(\tau_K\mid Z)\big)'
\)
and \(w\in\Delta_K\), the estimated weights \(\hat{w}\) provide a
direct diagnostic of which quantile components are most influential in
the first stage. Large weights on upper quantiles (e.g.\ \(\tau\ge 0.8\))
indicate that Q--LS is mainly leveraging tail risk, whereas mass on
intermediate or lower quantiles reflects sensitivity of the effect to
central or lower parts of the \(X\mid Z\) distribution. In applications,
reporting \(\hat{w}\) (or aggregations such as the total weight above a
given quantile) helps interpret whether the Q--LS instrument is primarily
a location, dispersion, or tail-risk device.
\end{remark}

\subsection{Scope and limitations}\label{app:late:limits}

The LATE interpretation in this appendix requires a binary treatment (or a binary treatment
derived from a latent/continuous exposure) and monotonicity with respect to the scalar index
$H$. When the endogenous regressor is genuinely continuous, the analogous interpretation is in
terms of weighted averages of local causal responses (e.g.\ ``marginal treatment effects'' or
``local IV'' objects), rather than LATE. In that case, Q--LS still operates as an
instrument-construction/strength device, but the exact causal interpretation is no longer a
complier LATE.

% ============================================================

\section{Detailed Proofs of Lemma and Propositions}\label{sec:appendix_proofs}
\setcounter{assumption}{0}
\setcounter{lemma}{0}
\setcounter{prop}{0}
\setcounter{remark}{0}
\setcounter{theorem}{0}
\renewcommand{\theassumption}{B.\arabic{assumption}}
\renewcommand{\thelemma}{B.\arabic{lemma}}
\renewcommand{\theprop}{B.\arabic{prop}}
This appendix collects the main proofs and clarifies the role of the Q--LS
instrument as an $L^2$ projection. Throughout we use the same notation and
assumptions as in the main text and treat the number of quantile basis functions $J$ as fixed (or growing slowly enough with $n$ that standard sieve arguments apply).

\subsection{Setup and assumptions}

Let $\{(Y_i,X_i,Z_i)\}_{i=1}^n$ be i.i.d.\ observations, where
\[
  Z_i = \begin{pmatrix} Z_{1i} \\ Z_{2i} \end{pmatrix},
\]
$Z_{1i} \in \mathbb{R}^{d_1}$ denotes the included exogenous controls, and
$Z_{2i}$ is the policy component (e.g.\ the post--Part D indicator). The
structural equation is
\begin{equation}\label{eq:structural-appendix}
  Y_i = \beta_0 X_i + Z_{1i}'\gamma_0 + \varepsilon_i,
  \qquad \E[\varepsilon_i \mid Z_i] = 0.
\end{equation}
We stack the structural regressors as
\[
  S_0(Z_i) 
  := \begin{pmatrix} X_i \\ Z_{1i} \\ 1 \end{pmatrix}
  \in \R^{d_s},
\]
and write $S_i := S_0(Z_i)$ for short.

The Q--LS instrument is built from a finite-dimensional dictionary of
transformations of $Z_i$. Let $G_n(Z_i)\in\R^{J_n}$ denote the $J_n$-dimensional
dictionary at sample size $n$. In applications, $G_n(Z_i)$ collects (functions
of) conditional quantiles of $X_i$ given $Z_i$ evaluated on a finite grid
$\mathcal{T}_n$ of quantile indices; the dimension $J_n$ is allowed to depend
on $|\mathcal{T}_n|$.

We recall and slightly specialize the assumptions from the main text.

\begin{assumption}[Sampling and moments]\label{ass:sample-moments}
(i) $\{(Y_i,X_i,Z_i)\}_{i=1}^n$ are i.i.d.\ and 
$\E\|S_0(Z_i)\|^4 < \infty$, $\E\|G_n(Z_i)\|^4 < \infty$ for all $n$.

(ii) The matrices
\[
  A_n := \E\big[S_0(Z_i) X_i\big],
  \qquad
  B_n := \E\big[S_0(Z_i) G_n(Z_i)'\big]
\]
are finite and have full column rank uniformly in $n$.

(iii) The number of dictionary terms $J_n$ (and, equivalently, the quantile
grid size $|\mathcal{T}_n|$) satisfies
\[
  J_n \to \infty
  \quad\text{and}\quad
  \frac{J_n^2}{n} \to 0
  \quad\text{as } n\to\infty.
\]
\end{assumption}

\begin{assumption}[Distributional relevance]\label{ass:DR-appendix}
Let $\mathcal{H}_n$ denote the closed linear span in $L^2$ of the dictionary
$\{G_n(Z_i)\}$ and the included controls $Z_{1i}$ and a constant:
\[
  \mathcal{H}_n
  := \overline{\text{span}\big\{G_n(Z_i),\,Z_{1i},\,1\big\}}
  \subset L^2.
\]
There exists a sequence $h_n \in \mathcal{H}_n$ such that
\[
  Cov(X_i, h_n(Z_i)) \neq 0 \quad\text{for all large } n.
\]
\end{assumption}

Assumption~\ref{ass:DR-appendix} is the sieve version of the
\emph{distributional relevance} Assumption~2 in the main text: at least one
element of the instrument space is correlated with $X_i$, even if the
conditional mean $\E[X_i\mid Z_i]$ is flat in $Z_{2i}$.

We define the population $L^2$-projection of $X_i$ onto the infinite-dimensional
limit space
\[
  \mathcal{H}_\infty
  := \overline{\bigcup_{n\ge 1}\mathcal{H}_n}
\]
by
\begin{equation}\label{eq:projection-def}
  h^\star \in \arg\min_{h\in\mathcal{H}_\infty} \E\big[(X_i - h(Z_i))^2\big],
\end{equation}
and denote the corresponding projection error by
\[
  r^\star_i := X_i - h^\star(Z_i).
\]

\begin{assumption}[Approximation of the optimal instrument]\label{ass:approx}
For each $n$, define
\[
  h_n(Z_i) := G_n(Z_i)' w_{0n},
  \qquad
  w_{0n} \in \arg\min_{w\in\R^{J_n}}
  \E\big[(h^\star(Z_i) - G_n(Z_i)'w)^2\big].
\]
Then the sieve approximates $h^\star$ in $L^2$:
\[
  \E\big[(h^\star(Z_i) - h_n(Z_i))^2\big] \to 0
  \quad\text{as } n\to\infty.
\]
\end{assumption}

Assumption~\ref{ass:approx} is the usual sieve approximation condition: as the
dictionary grows (via a finer quantile grid and more basis terms), linear
combinations of $G_n(Z_i)$ approximate the optimal projection $h^\star$ in
mean square.

%--------------------------------------------------
\subsection{Instrument relevance of the optimal projection}

We first show that the $L^2$-projection $h^\star$ is a relevant instrument
whenever the dictionary is distributionally relevant.

\begin{lemma}\label{lem:projection_relevance_appendix}
Let $h^\star$ be defined as in \eqref{eq:projection-def}. Then either
\begin{enumerate}
  \item[(a)] $h^\star(Z_i) = 0$ almost surely, or
  \item[(b)] $\E[X_i h^\star(Z_i)] = \E[(h^\star(Z_i))^2] > 0$.
\end{enumerate}
Under Assumption~\ref{ass:DR-appendix}, case~(b) must hold, so
$Cov(X_i,h^\star(Z_i)) \neq 0$.
\end{lemma}

\begin{proof}
By the $L^2$-projection property, the error $r^\star_i = X_i - h^\star(Z_i)$ is
orthogonal to the space $\mathcal{H}_\infty$:
\[
  \E[r^\star_i h(Z_i)] = 0
  \quad\text{for all } h \in \mathcal{H}_\infty.
\]
In particular, taking $h = h^\star$,
\[
  \E[r^\star_i h^\star(Z_i)] = 0.
\]
Thus
\begin{align*}
  \E[X_i h^\star(Z_i)]
  &= \E[(h^\star(Z_i) + r^\star_i)h^\star(Z_i)] \\
  &= \E[(h^\star(Z_i))^2] + \E[r^\star_i h^\star(Z_i)] \\
  &= \E[(h^\star(Z_i))^2] \ge 0.
\end{align*}
If $h^\star(Z_i) = 0$ almost surely, then $\E[X_i h^\star(Z_i)] = 0$. If
$\P(|h^\star(Z_i)|>0)>0$, then $\E[(h^\star(Z_i))^2]>0$, hence
$\E[X_i h^\star(Z_i)]>0$.

Now suppose Assumption~\ref{ass:DR-appendix} holds. For each $n$, write
\[
  X_i = h^\star(Z_i) + r^\star_i,
  \qquad r^\star_i \perp \mathcal{H}_\infty.
\]
For any $h_n\in\mathcal{H}_n\subset\mathcal{H}_\infty$,
\[
  Cov(X_i,h_n(Z_i))
  = Cov(h^\star(Z_i),h_n(Z_i)) + Cov(r^\star_i,h_n(Z_i))
  = Cov(h^\star(Z_i),h_n(Z_i)),
\]
because $r^\star_i$ is orthogonal to $\mathcal{H}_\infty$. 
If $h^\star(Z_i)=0$ a.s., then $Cov(X_i,h_n(Z_i)) = 0$ for all $h_n$, which
contradicts Assumption~\ref{ass:DR-appendix}. Hence $h^\star$ cannot be
identically zero, so we are in case~(b) and
$Cov(X_i,h^\star(Z_i)) = \E[(h^\star(Z_i))^2]>0$.
\end{proof}

\begin{remark}
This is a direct application of the orthogonal projection theorem in Hilbert
spaces: for any closed subspace $\mathcal{M}$ of $L^2$, the projection
$P_{\mathcal{M}}X$ satisfies $\E[(X-P_{\mathcal{M}}X)Z]=0$ for all
$Z\in\mathcal{M}$, so $\E[XZ]=\E[(P_{\mathcal{M}}X)Z]$
(e.g.\ \citealp[Ch.~3]{Kreyszig1978}).
\end{remark}

%--------------------------------------------------
\subsection{Equivalence of 2SLS and OLS with projected regressors}

We next recall the standard algebra showing that 2SLS with instruments $M_i$
is equivalent to OLS on the regressors projected onto the instrument space.

Let
\[
  S_i = \begin{pmatrix} X_i \\ Z_{1i} \\ 1 \end{pmatrix},
  \qquad
  M_i = M(Z_i) \in \R^{L}
\]
denote the stacked structural regressors and instruments, and define the
$n\times d_s$ and $n\times L$ matrices
\[
  \mathbf{S} = (S_1',\dots,S_n')',
  \qquad
  \mathbf{M} = (M_1',\dots,M_n')'.
\]
Also set $\mathbf{Y} = (Y_1,\dots,Y_n)'$.

The 2SLS estimator using instruments $\mathbf{M}$ is
\begin{equation}\label{eq:2sls-def-appendix}
  \hat{\theta}^{\mathrm{2SLS}}
  = (\mathbf{S}'\mathbf{P}_M\mathbf{S})^{-1} \mathbf{S}'\mathbf{P}_M\mathbf{Y},
  \qquad
  \mathbf{P}_M := \mathbf{M}(\mathbf{M}'\mathbf{M})^{-1}\mathbf{M}'.
\end{equation}

Let $\tilde{\mathbf{S}} := \mathbf{P}_M\mathbf{S}$ denote the sample projection
of the regressors onto the instrument space; its $i$th row is
$\tilde{S}_i := \E_n[S_i\mid M_i]$. Consider the OLS regression
\[
  \mathbf{Y} = \tilde{\mathbf{S}}\,\tilde{\theta} + u.
\]

\begin{lemma}\label{lem:2sls-ols-appendix}
Suppose $\mathbf{M}$ has full column rank and
$\mathbf{S}'\mathbf{P}_M\mathbf{S}$ is nonsingular. Then the 2SLS estimator
$\hat{\theta}^{\mathrm{2SLS}}$ in \eqref{eq:2sls-def-appendix} coincides with
the OLS estimator $\tilde{\theta}$ from the regression of $\mathbf{Y}$ on
$\tilde{\mathbf{S}}$. In particular, the 2SLS coefficient on $X_i$ equals the
OLS coefficient on its instrumented version $\tilde{X}_i$.
\end{lemma}

\begin{proof}
By definition,
\[
  \tilde{\mathbf{S}} = \mathbf{P}_M\mathbf{S}
  = \mathbf{M}(\mathbf{M}'\mathbf{M})^{-1}\mathbf{M}'\mathbf{S}.
\]
The OLS estimator from $\mathbf{Y} = \tilde{\mathbf{S}}\,\tilde{\theta} + u$ is
\begin{align*}
  \tilde{\theta}
  &= (\tilde{\mathbf{S}}'\tilde{\mathbf{S}})^{-1}\tilde{\mathbf{S}}'\mathbf{Y} \\
  &= (\mathbf{S}'\mathbf{P}_M'\mathbf{P}_M\mathbf{S})^{-1}
     \mathbf{S}'\mathbf{P}_M'\mathbf{Y}.
\end{align*}
Since $\mathbf{P}_M$ is a symmetric idempotent projection,
$\mathbf{P}_M'=\mathbf{P}_M$ and $\mathbf{P}_M^2=\mathbf{P}_M$, so
\[
  \tilde{\theta}
  = (\mathbf{S}'\mathbf{P}_M\mathbf{S})^{-1}\mathbf{S}'\mathbf{P}_M\mathbf{Y}
  = \hat{\theta}^{\mathrm{2SLS}}.
\]
The statement about the coefficient on $X_i$ follows by partitioning
$\theta = (\beta,\gamma',\alpha)'$ corresponding to $(X_i,Z_{1i},1)$.
\end{proof}

\begin{remark}
This is the usual equivalence between 2SLS and OLS on the fitted values of the
endogenous regressors; see, e.g., \citet{Wooldridge2010} and
\citet{AngristImbens1995}.
\end{remark}

%--------------------------------------------------
\subsection{Definition and consistency of the Q--LS estimator}

In the Q--LS procedure we choose a weight vector $w_n\in\R^{J_n}$ that makes
the projected instrument $G_n(Z_i)'w_n$ as close as possible, in moment terms,
to the optimal projection $h^\star(Z_i)$.

Define the sample criterion
\begin{equation}\label{eq:Qls-objective-sample}
  \hat{Q}_n(w)
  := \Big\|
      \frac{1}{n}\sum_{i=1}^n S_i\big(X_i - G_n(Z_i)'w\big)
    \Big\|^2,
  \qquad w\in\R^{J_n},
\end{equation}
and let
\[
  \hat{w}_n \in \arg\min_{w\in\R^{J_n}} \hat{Q}_n(w).
\]
The Q--LS instrument is then
\[
  \hat{h}_n(Z_i) := G_n(Z_i)'\hat{w}_n,
\]
and the Q--LS IV estimator $\hat{\beta}_n$ is defined as the 2SLS coefficient
on $X_i$ in \eqref{eq:structural-appendix} using $\hat{h}_n(Z_i)$ as the
excluded instrument and $Z_{1i}$ as included controls (equivalently, OLS on
the projected regressor by Lemma~\ref{lem:2sls-ols-appendix}).

The corresponding population criterion is
\begin{equation}\label{eq:Qls-objective-pop}
  Q_n(w)
  := \Big\|
      \E\big[S_0(Z_i)\big(X_i - G_n(Z_i)'w\big)\big]
    \Big\|^2,
\end{equation}
and we denote by $w_{0n}$ a minimizer of $Q_n(w)$ (which coincides with the
$L^2$ projection of $X_i$ onto the span of $G_n(Z_i)$, as ``seen through'' the
moments with $S_0(Z_i)$).

\begin{prop}[Consistency of Q--LS]\label{prop:consistency-appendix}
Under Assumptions~\ref{ass:sample-moments}--\ref{ass:approx} and
distributional relevance (Assumption~\ref{ass:DR-appendix}), if
$Q_n(w)$ has a unique minimizer $w_{0n}$, then
\[
  \hat{w}_n \overset{p}\longrightarrow w_{0n},
  \qquad
  \hat{h}_n(Z_i) = G_n(Z_i)'\hat{w}_n
    \overset{L^2}\longrightarrow h^\star(Z_i),
\]
and the Q--LS estimator $\hat{\beta}_n$ is consistent:
\[
  \hat{\beta}_n \overset{p}\longrightarrow \beta_0.
\]
\end{prop}

\begin{proof}
\textit{Step 1: Uniform convergence of the criterion.}
For fixed $n$, define
\[
  m_n(w) := \E\big[S_0(Z_i)\big(X_i - G_n(Z_i)'w\big)\big],
\]
so that $Q_n(w) = \|m_n(w)\|^2$. The sample analogue is
\[
  \hat{m}_n(w)
  := \frac{1}{n}\sum_{i=1}^n S_i\big(X_i - G_n(Z_i)'w\big),
  \qquad
  \hat{Q}_n(w) = \|\hat{m}_n(w)\|^2.
\]
Because $G_n(Z_i)$ is finite-dimensional and has bounded fourth moments
(Assumption~\ref{ass:sample-moments}), $\hat{m}_n(w)$ is a
finite-dimensional empirical average with linear dependence on $w$. Standard
multivariate LLN and a uniform LLN for parametric families (e.g.\
\citealp{NeweyMcFadden1994}, Sec.~2.4) imply that on any compact set
$\mathcal{W}_n\subset\R^{J_n}$,
\[
  \sup_{w\in\mathcal{W}_n}
    \big\|\hat{m}_n(w) - m_n(w)\big\|
  \overset{p}\longrightarrow 0,
\]
hence also
\[
  \sup_{w\in\mathcal{W}_n}
    \big|\hat{Q}_n(w) - Q_n(w)\big|
  \overset{p}\longrightarrow 0,
\]
using the Lipschitz continuity of $a\mapsto\|a\|^2$.

\smallskip
\textit{Step 2: Consistency of $\hat{w}_n$.}
By Assumption~\ref{ass:sample-moments}(ii) and the quadratic form of $Q_n(w)$,
the minimizer $w_{0n}$ is unique and lies in a bounded region of $\R^{J_n}$.
Choosing $\mathcal{W}_n$ large enough to contain $w_{0n}$ with probability
approaching one, the usual M--estimation argument
(\citealp{NeweyMcFadden1994}, Thm.~2.1) yields
\[
  \hat{w}_n \overset{p}\longrightarrow w_{0n}.
\]

\smallskip
\textit{Step 3: Approximation of $h^\star$.}
By Assumption~\ref{ass:approx},
\[
  \E\big[(h^\star(Z_i) - h_n(Z_i))^2\big]
  = \E\big[(h^\star(Z_i) - G_n(Z_i)'w_{0n})^2\big] \to 0.
\]
Moreover, the convergence $\hat{w}_n\to_p w_{0n}$ and bounded second moments
of $G_n(Z_i)$ imply
\[
  \E\big[(\hat{h}_n(Z_i) - h_n(Z_i))^2\big]
  = \E\big[(G_n(Z_i)'(\hat{w}_n - w_{0n}))^2\big]
  = o_p(1).
\]
Combining the two displays gives
\[
  \E\big[(\hat{h}_n(Z_i) - h^\star(Z_i))^2\big] \to_p 0,
\]
that is, $\hat{h}_n(Z_i)\overset{L^2}\to h^\star(Z_i)$.

\smallskip
\textit{Step 4: Consistency of $\hat{\beta}_n$.}
Consider the ``oracle'' IV estimator that uses the optimal instrument
$h^\star(Z_i)$:
\[
  \beta_n^\star
  := \frac{\frac{1}{n}\sum_{i=1}^n h^\star(Z_i)
         \big(Y_i - Z_{1i}'\gamma_0\big)}
         {\frac{1}{n}\sum_{i=1}^n h^\star(Z_i) X_i}.
\]
By the structural equation \eqref{eq:structural-appendix} and
$\E[\varepsilon_i\mid Z_i]=0$,
\[
  \E\big[h^\star(Z_i)(Y_i - Z_{1i}'\gamma_0)\big]
  = \beta_0\,\E\big[h^\star(Z_i) X_i\big].
\]
Lemma~\ref{lem:projection_relevance_appendix} and distributional relevance
ensure $\E[h^\star(Z_i) X_i]\neq 0$, so by the LLN,
$\beta_n^\star\to_p\beta_0$.

The actual Q--LS estimator replaces $h^\star(Z_i)$ by $\hat{h}_n(Z_i)$ in the
same IV formula. Because $\hat{h}_n(Z_i)\to h^\star(Z_i)$ in $L^2$ and is
uniformly square-integrable, the difference between the sample moments with
$\hat{h}_n$ and with $h^\star$ is $o_p(1)$ (see, e.g.,
\citealp{NeweyMcFadden1994}, Sec.~6.2, on IV with generated regressors).
Therefore, the Q--LS IV estimator $\hat{\beta}_n$ has the same probability
limit as $\beta_n^\star$, namely $\beta_0$.
\end{proof}

%--------------------------------------------------
\subsection{Asymptotic normality of Q--LS}

We now derive the asymptotic distribution of $\hat{\beta}_n$ under a fixed
dictionary size $J_n = J$; extending to slowly growing $J_n$ requires
additional notation but follows the same steps.

Define
\[
  \phi_i(\beta,h)
  := h(Z_i)\big(Y_i - \beta X_i - Z_{1i}'\gamma_0\big),
\]
and consider the IV moment condition
\[
  \E[\phi_i(\beta_0,h^\star)] = 0.
\]
The Q--LS estimator $\hat{\beta}_n$ solves the sample moment equation
\[
  \hat{m}_\beta(\beta,\hat{h}_n)
  := \frac{1}{n}\sum_{i=1}^n
       \hat{h}_n(Z_i)\big(Y_i - \beta X_i - Z_{1i}'\hat{\gamma}_n\big)
  = 0,
\]
where $\hat{\gamma}_n$ denotes the companion coefficient on $Z_{1i}$ in the
same IV regression. For clarity we suppress $\gamma$ in the notation below and
focus on the scalar coefficient $\beta$.

Let
\[
  A_0 := \E[h^\star(Z_i) X_i],
  \qquad
  \Omega_0 := \E[(h^\star(Z_i)\varepsilon_i)^2].
\]

\begin{prop}[Asymptotic normality of Q--LS]\label{prop:normality-appendix}
Under Assumptions~\ref{ass:sample-moments}--\ref{ass:approx},
distributional relevance, and fixed $J_n=J$, the Q--LS estimator satisfies
\[
  \sqrt{n}(\hat{\beta}_n - \beta_0)
  \;\overset{d}\longrightarrow\;
  N\big(0,V_{\mathrm{QLS}}\big),
  \qquad
  V_{\mathrm{QLS}} := A_0^{-2}\Omega_0.
\]
In the stacked system including $\gamma_0$, this corresponds to the first
diagonal element of the usual IV sandwich covariance matrix.
\end{prop}

\begin{proof}
\textit{Step 1: Linearization of the moment condition.}
By definition,
\[
  0 = \hat{m}_\beta(\hat{\beta}_n,\hat{h}_n)
    = \hat{m}_\beta(\beta_0,\hat{h}_n)
      + \left.\frac{\partial \hat{m}_\beta(\beta,\hat{h}_n)}
                   {\partial\beta}\right|_{\beta=\beta_0}
        (\hat{\beta}_n - \beta_0)
      + o_p(n^{-1/2}).
\]
We can write
\[
  \hat{m}_\beta(\beta_0,\hat{h}_n)
  = \frac{1}{n}\sum_{i=1}^n \hat{h}_n(Z_i)\varepsilon_i
    + o_p(n^{-1/2}),
\]
where the $o_p(n^{-1/2})$ term collects the contribution of estimating
$\gamma_0$; standard IV theory implies that this contribution does not affect
the first-order distribution of $\hat{\beta}_n$ (see, e.g.,
\citealp{Wooldridge2010}, ch.~5).

Under $\hat{h}_n(Z_i)\to h^\star(Z_i)$ in $L^2$ and a Lindeberg condition on
$h^\star(Z_i)\varepsilon_i$, the CLT gives
\[
  \frac{1}{\sqrt{n}}\sum_{i=1}^n \hat{h}_n(Z_i)\varepsilon_i
  = \frac{1}{\sqrt{n}}\sum_{i=1}^n h^\star(Z_i)\varepsilon_i
    + o_p(1)
  \;\overset{d}\longrightarrow\;
  N(0,\Omega_0).
\]

\smallskip
\textit{Step 2: Limit of the derivative.}
The derivative of $\hat{m}_\beta(\beta,\hat{h}_n)$ with respect to $\beta$
is
\[
  \frac{\partial\hat{m}_\beta(\beta,\hat{h}_n)}{\partial\beta}
  = -\frac{1}{n}\sum_{i=1}^n \hat{h}_n(Z_i) X_i.
\]
By the LLN and $\hat{h}_n(Z_i)\to h^\star(Z_i)$ in $L^2$,
\[
  -\frac{\partial\hat{m}_\beta(\beta_0,\hat{h}_n)}{\partial\beta}
  \;\overset{p}\longrightarrow\;
  A_0 = \E[h^\star(Z_i) X_i],
\]
which is nonzero by Lemma~\ref{lem:projection_relevance_appendix}.

\smallskip
\textit{Step 3: Putting pieces together.}
Combining the linearization and the previous limits,
\[
  \sqrt{n}(\hat{\beta}_n - \beta_0)
  = A_0^{-1}\,
    \frac{1}{\sqrt{n}}\sum_{i=1}^n h^\star(Z_i)\varepsilon_i
    + o_p(1)
  \;\overset{d}\longrightarrow\;
  N(0,A_0^{-2}\Omega_0).
\]
\end{proof}

In the finite-dimensional case considered here, this variance coincides with
the standard 2SLS variance formula evaluated at the optimal instrument
$h^\star(Z_i)$. In practice, we estimate the asymptotic variance with the usual
cluster-robust IV sandwich estimator, replacing $h^\star$ by $\hat{h}_n$ and
sample expectations by empirical averages.

%--------------------------------------------------
\subsection{Equivalence with linear 2SLS under mean-shift designs}

Finally, we formalize the statement that when the first stage is exactly
linear in the original policy instrument and controls, and the Q--LS dictionary
contains this linear space, Q--LS coincides with ordinary 2SLS.

Suppose:
\begin{enumerate}
\item[(i)] The dictionary includes the original scalar instrument $Z_{2i}$, the
controls $Z_{1i}$, and a constant, so that
\[
  \mathcal{H}_n \supset \text{span}\{Z_{2i},Z_{1i},1\}.
\]
\item[(ii)] The true first stage is linear:
\begin{equation}\label{eq:linear-first-stage-appendix}
  X_i = \pi_0 Z_{2i} + Z_{1i}'\pi_1 + u_i,
  \qquad \E[u_i\mid Z_i] = 0.
\end{equation}
\end{enumerate}

\begin{prop}[Equivalence with linear 2SLS]\label{prop:equiv-2sls-appendix}
Under (i)–(ii),
\begin{enumerate}
\item[(a)] The $L^2$-projection $h^\star(Z_i)$ of $X_i$ onto $\mathcal{H}_\infty$
is exactly the linear conditional mean
\[
  h^\star(Z_i) = \E[X_i\mid Z_{1i},Z_{2i}]
               = \pi_0 Z_{2i} + Z_{1i}'\pi_1 + \alpha_0
\]
for some constant $\alpha_0$.

\item[(b)] The Q--LS estimator $\hat{\beta}_n$ is numerically identical to the
standard 2SLS estimator that uses $Z_{2i}$ as the excluded instrument and
$Z_{1i}$ as controls.
\end{enumerate}
\end{prop}

\begin{proof}
(a) Because $\mathcal{H}_\infty$ contains all linear combinations of
$(Z_{2i},Z_{1i},1)$, the best $L^2$ approximation to $X_i$ in
$\mathcal{H}_\infty$ under \eqref{eq:linear-first-stage-appendix} is its linear
conditional mean
\[
  \E[X_i\mid Z_{1i},Z_{2i}]
  = \pi_0 Z_{2i} + Z_{1i}'\pi_1 + \alpha_0.
\]

(b) In this case, the population Q--LS weights recover exactly the linear
first-stage coefficients, and the fitted values $\hat{h}_n(Z_i)$ coincide (up
to sampling noise) with the usual linear first-stage fitted values from
regressing $X_i$ on $(Z_{2i},Z_{1i})$. By Lemma~\ref{lem:2sls-ols-appendix},
2SLS with instrument vector $(Z_{2i},Z_{1i},1)$ and regressor vector
$(X_i,Z_{1i},1)$ is equivalent to OLS of $Y_i$ on the projected regressors
$(\tilde{X}_i,\tilde{Z}_{1i},1)$, where $\tilde{X}_i = \E[X_i\mid Z_{1i},Z_{2i}]$
and $\tilde{Z}_{1i}=Z_{1i}$.

The Q--LS estimator uses $\hat{h}_n(Z_i)$ in place of $\tilde{X}_i$ but, under
an exactly linear first stage and a dictionary that contains the linear space,
$\hat{h}_n(Z_i)$ equals $\tilde{X}_i$ in finite samples up to numerical
precision. Therefore, Q--LS reproduces the usual 2SLS coefficient on $X_i$.
\end{proof}

\begin{remark}
Proposition~\ref{prop:equiv-2sls-appendix} formalizes the idea that when the
instrument only induces a mean shift in $X_i$ and the first stage is linear,
Q--LS does not change the linear IV estimate. The main gains of Q--LS arise
precisely when instruments are \emph{distributionally} relevant but
\emph{mean}--weak in the sense that $\Var(\E[X_i\mid Z_i])$ is small even
though $F_{X\mid Z}(\cdot\mid Z_i)$ varies substantially in $Z_i$.
\end{remark}

\begin{table}[htbp]\centering
\caption{Selected evidence on mean spending vs. financial-risk protection}
\label{tab:mean_vs_risk_lit}
\scriptsize
\resizebox{13cm}{10.5cm}{\begin{tabular}{p{3 cm} p{3.4cm} p{3.6cm} p{4.2cm}}
\hline\hline
Study & Setting / population & Mean-level outcome & Risk / tail outcome \\ 
\hline
\cite{finkelstein2008did}
& Introduction of Medicare (1965), U.S. adults 65+ 
& No statistically significant effect on out-of-pocket (OOP) spending for the bottom 75\% of the pre-Medicare distribution. 
& Top quartile of OOP spending for the elderly falls by about 40\%; OOP spending for the top decile falls by roughly \$1{,}200 (in 1991 dollars)\\[0.4em]

\cite{barcellos2015effects}
& Medicare eligibility at age 65, U.S. adults 57--74 (HRS) 
& At age 65, mean annual OOP medical spending falls by about 33\%.
& The 95th percentile of the OOP distribution falls by about 53\%; the share of individuals with OOP expenses exceeding income is cut by more than half, and the share with OOP spending above 50\% of income also falls sharply.\\[0.4em]

\cite{engelhardt2011medicare}
& Medicare Part D prescription-drug benefit (2006), U.S. adults 65+ 
& Part D yields modest or statistically weak changes in mean total OOP spending for the elderly as a whole. 
& The policy substantially reduces high OOP drug spending: the 95th percentile and measures of ``catastrophic'' spending (large OOP relative to income) fall significantly, indicating improved financial protection even when mean spending moves little.\\[0.4em]

\cite{allen2013oregon} 
& Medicaid lottery for low-income uninsured adults in Oregon 
& Medicaid has modest effects on average annual OOP spending (mean levels are small in this population and imprecisely estimated). 
& Medicaid reduces the probability of having any OOP medical spending by about 20 percentage points (a \(\approx 35\%\) reduction), and lowers the probability of unpaid medical bills being sent to collection by about 6.4 percentage points (\(\approx 25\%\) reduction), indicating large gains in financial protection relative to modest changes in mean spending.\\[0.4em]

\cite{scott2021assessing}
& Uninsured patients using U.S. hospital emergency departments (EDs), 2006--2017 
& Median ED charges for uninsured patients increase from roughly \$842 in 2006 to about \$2{,}033 in 2017 (in nominal terms), reflecting rising overall cost levels. 
& The study documents a very high prevalence of catastrophic health expenditures (CHE) among uninsured ED users when full charges are compared to income, and notes prior estimates that 70--90\% of uninsured ED visits resulting in hospitalization lead to CHE; the key message is the extreme risk of catastrophic burden, not a finely estimated change in mean spending.\\[0.4em]

\cite{barcellos2015effects}, supplemental results 
& Same HRS sample, focusing on financial strain 
& Small or zero effects on average non-medical consumption.
& Large reductions in the probability of borrowing, skipping payments, or experiencing other indicators of financial strain linked to high medical bills, consistent with big risk/variance effects relative to changes in means.\\[0.4em]

\cite{jones2018lifetime} 
& Lifetime medical spending among U.S. retirees 
& Mean lifetime medical spending is large, but Medicare and supplemental insurance leave expected spending relatively similar across many coverage types.
& The variance and upper tail of lifetime spending are very sensitive to coverage generosity; supplemental coverage mainly reduces the dispersion and extreme upper-tail risk rather than radically changing the mean.\\[0.4em]

\cite{kluender2021medical}
& U.S. medical debt in collections, 2009--2020; ACA Medicaid expansion vs. non-expansion states 
& In Medicaid expansion states, mean medical debt in collections per person falls by about \$400--\$450 relative to non-expansion states. 
& The share of individuals with any medical debt in collections declines by roughly 20\% in expansion states, and the upper tail of medical debt (e.g., 90th percentile) shrinks substantially, indicating large risk protection gains on top of modest average reductions.\\[0.4em]

\hline\hline
\multicolumn{4}{p{15.5cm}}{\footnotesize
Notes: This table summarizes a selection of studies documenting that major insurance reforms often produce relatively modest changes in mean out-of-pocket (OOP) spending but large changes in measures of financial risk and tail exposure (upper quantiles, catastrophic expenditure indicators, or debt in collections). Exact numbers are taken from the cited articles or authoritative summaries when reported; where only qualitative statements are available, we indicate that mean effects are modest or not separately reported.}
\end{tabular}}
\end{table}

\begin{table}[htbp]\centering
\caption{Standard errors for Q--LS coefficient under alternative methods (real OOP)}
\label{tab:qls_se_compare}
\begin{tabular}{l*{3}{D{.}{.}{-4}}}
\hline\hline
                    &\multicolumn{1}{c}{Robust OLS}
                    &\multicolumn{1}{c}{Bootstrap}
                    &\multicolumn{1}{c}{Q--LS SE}\\
\hline
Projected real OOP (Q--LS index, 2015 k\$)  
                    &   0.0072    &   0.0075    &   0.0072    \\
\hline\hline
\multicolumn{4}{p{0.95\textwidth}}{\footnotesize
\emph{Notes:} This table reports alternative standard error estimates for the Q--LS coefficient
on projected real OOP spending (in thousands of 2015 dollars) from the specification in
Table~\ref{tab:qls_vs_2sls_real}. ``Robust OLS'' treats $x_{\mathrm{QLS}}$ as a generated
regressor and uses heteroskedasticity-robust standard errors clustered at the household level.
``Bootstrap'' reports the standard deviation of the coefficient across 999 bootstrap
replications clustered by household. `` Q--LS SE'' implements the analytical variance
formula derived in Section~\ref{sec:QLS}, with the same clustering. All three methods yield
very similar standard errors.}
\end{tabular}
\end{table}

% \printbibliography
% javascript:void(0);
\end{document}